\documentclass[sigconf]{acmart}
\renewcommand\footnotetextcopyrightpermission[1]{}
\setcopyright{none}
\acmConference{}{}{}
\acmISBN{}
\acmDOI{}
\acmPrice{}
\acmYear{}
\copyrightyear{}

\usepackage[framemethod=TikZ]{mdframed}

\usepackage{multicol}
\usepackage{multirow}
\usepackage{tikz}
\usepackage{amsmath}
\usepackage{enumitem} 
\usepackage{mathrsfs}
\usepackage{soul}
\usepackage{filecontents}
\usepackage{amsthm}
\usepackage{booktabs}
\usepackage{threeparttable}
\usepackage{cleveref}

\newtheorem{lemma}{Lemma}[section]
\crefname{lemma}{Lemma}{Lemmas}
\Crefname{lemma}{Lemma}{Lemmas}
\newtheorem{proposition}{Proposition}
\crefname{proposition}{Proposition}{Propositions}
\Crefname{proposition}{Proposition}{Propositions}
\newtheorem{corollary}{Corollary}
\crefname{corollary}{Corollary}{Corollaries}
\Crefname{corollary}{Corollary}{Corollaries}
\newtheorem{observation}{Observation}
\crefname{observation}{Observation}{Observations}
\Crefname{observation}{Observation}{Observations}
\newtheorem{assumption}{Assumption}
\crefname{assumption}{Assumption}{Assumptions}
\Crefname{assumption}{Assumption}{Assumptions}
\begin{document}

\title{Ordering Power is Sanctioning Power: Sanction-Evasion MEV and \\the Limits of On-Chain Enforcement}

\author{Di Wu}
\authornote{Work done at both Zhejiang University and The Hong Kong University of Science and Technology (Guangzhou).}
\affiliation{%
  \institution{Zhejiang University}
  \country{}
}
\affiliation{%
  \institution{ The Hong Kong University of Science and Technology (Guangzhou)}
  \country{}
}
\email{wu.di@zju.edu.cn}

\author{Yuman Bai}
\affiliation{\institution{Zhejiang University}\country{}}
\email{baiyuman@zju.edu.cn}

\author{Shoupeng Ren}
\affiliation{\institution{Zhejiang University}\country{}}
\email{spren@zju.edu.cn}

\author{Xinyu Zhang}
\affiliation{\institution{Sun Yat-Sen University}\country{}}
\email{zhangxy353@mail2.sysu.edu.cn}

\author{Yiyue Cao}
\affiliation{\institution{The Hong Kong University of Science and Technology (Guangzhou)}\country{}}
\email{ycao948@connect.hkust-gz.edu.cn}

\author{Xuechao Wang}
\authornote{Corresponding author.}
\affiliation{\institution{The Hong Kong University of Science and Technology (Guangzhou))}\country{}}
\email{xuechaowang@hkust-gz.edu.cn}

\author{Wu Wen}
\affiliation{\institution{Zhejiang University}\country{}}
\email{wu.wen@intl.zju.edu.cn}

\author{Jian Liu}
\affiliation{\institution{Zhejiang University}\country{}}
\email{liujian2411@zju.edu.cn}

\begin{abstract}

Centralized stablecoins such as USDT and USDC enforce financial sanctions through contract-layer blacklist functions, yet on a public blockchain the freeze itself is only an ordinary transaction that must compete with the sanctioned party's transfer for execution priority.
We identify a structural gap between contract-layer authority and ordering-layer enforcement in sanction races: when the two transactions race for inclusion in the same block, the outcome is decided not by regulatory mandate but by the economically motivated ordering choices of block producers.
Since either side can pay for priority, this contest induces a rent extracted by block producers, which we term \emph{Sanction-Evasion MEV (SE-MEV)}.
To quantify this gap in practice, we construct the first longitudinal dataset of on-chain sanctions enforcement and evasion for Ethereum-based USDT and USDC from November 2017 to August 2025, covering more than \$1.5 billion in frozen value. 
We find that at least 7.3\% of sanctioned USDT addresses and at least 18.7\% of sanctioned USDC addresses had already been drained to zero by the time the freeze took effect. We further document an escalation trajectory from issuer-side out-of-gas failures, to public gas auctions, to private order flow, and finally to direct payments to block producers. 
Together, these stages provide the first empirical evidence that block producers extract MEV from sanction enforcement. 
To understand the long-run consequences of this dynamic, we develop a game-theoretic model of stablecoin sanctions that explicitly incorporates MEV. 
The model yields three results: (i) a compliant issuer cannot rationally remain outside the ordering market; (ii) fixed participation costs concentrate evasion among a small set of specialized, MEV-aware adversaries; and (iii) the implicit MEV tax extracted by block producers grows without bound as regulatory penalties intensify, which in turn creates persistent incentives for issuers to vertically integrate into block-building infrastructure. 
This gap is not unique to stablecoin sanctions: the same conflict arises whenever a privileged on-chain action must be executed as an ordinary transaction, such as emergency pauses, governance interventions, or judicial freezes. Whenever ordering power is allocated by economic incentives, \textbf{ordering power is sanctioning power}, and contract-layer authority alone cannot guarantee enforcement.
\end{abstract}

\maketitle
\pagestyle{plain}

\newcommand{\Paragraph}[1]{\noindent \textbf{#1}}

\mdfdefinestyle{findingstyle}{%
  linecolor=black!70,
  linewidth=0.6pt,
  roundcorner=2pt,
  backgroundcolor=white,
  innertopmargin=2pt,
  innerbottommargin=2pt,
  innerleftmargin=8pt,
  innerrightmargin=8pt,
  skipabove=2pt,
  skipbelow=2pt,
}
\newmdtheoremenv[style=findingstyle]{finding}{Finding}


\section{Introduction}
Stablecoins have emerged as one of the most critical pieces of infrastructure within the cryptocurrency ecosystem, serving as the backbone for trading, lending, and cross-border payments~\cite{imf2025understanding}. 
As of April 2026, the total market capitalization of stablecoins had reached \$320 billion~\cite{DefillamaStablecoins}. 
Among them, centralized stablecoins, most notably USDT~\cite{tether} and USDC~\cite{circle}, dominate market share.
Their issuers can freeze assets at designated addresses through contract-layer blacklist functions. In prevailing industry and regulatory narratives, this capability is typically treated as a sufficient condition for effective sanctions enforcement: as long as the issuer possesses freeze authority, it can prevent a targeted address from transferring its assets further~\cite{genius-act-2025,mica-2023}. This intuition, however, does not automatically hold in the context of public blockchains, because a sanctions decision does not itself alter on-chain state.

A blacklist invocation is, in essence, an ordinary on-chain transaction. In networks such as Ethereum~\cite{wood2014ethereum}, where transaction ordering is allocated through economic incentives, under both Proof-of-Stake and Proof-of-Work, a sanctioned entity can submit a transfer transaction in the same block and compete directly with the freeze transaction for execution priority. This is not merely a hypothetical concern. We observe cases in which the issuer's freeze transaction entered the mempool earlier, yet the evasive transfer was confirmed on-chain first: the sanction was issued, but no target value was actually frozen.

We distinguish this phenomenon from broader cases of early fund movement, because not every transfer before a freeze raises an ordering-layer concern. If a target moves assets based on informational advantage or anticipated risk, the outcome is primarily explained by information and timing, rather than transaction ordering. Such cases fall outside the core problem studied in this paper. 
Our focus is instead on \emph{reactive} evasion, where the target's transfer and the issuer's freeze are triggered by the same external signal, such as a sanctions listing, a public hack attribution, or another publicly observable enforcement trigger. 
In this setting, both parties may simply react to the same off-chain event and submit competing transactions into the same ordering pipeline. At that point, enforcement no longer depends solely on whether the issuer possesses blacklist authority; it also depends on which transaction the ordering layer executes first. This raises a fundamental question: \textbf{when a freeze and a transfer compete directly for execution priority, who is effectively exercising sanctions enforcement power?}

The central claim of this paper is that, \textbf{in reactive sanctions races, contract-layer authority is only a necessary condition for enforcement. Whether sanctions are actually enforced depends on how the ordering layer allocates execution priority.} 
This reframes sanctions enforcement as an ordering-market problem: when execution priority can be influenced through incentives such as transaction fees, the right to determine whether a freeze or a transfer executes first becomes economically contestable. 
\emph{Maximal extractable value (MEV)}~\cite{daian2020flash} describes precisely this class of phenomena, where transaction-ordering infrastructure can profit from controlling execution priority. 
In reactive sanctions races, the confrontation between sanctioning and evading parties creates a profit opportunity for actors that can influence execution priority. 
We term this sanction-induced extractable value \emph{Sanction-Evasion MEV (SE-MEV)}.

To quantify this enforcement limitation induced by ordering layer competition systematically, we construct the first dataset of on-chain sanctions enforcement and evasion for Ethereum-based USDC and USDT from Nov 2017 to Aug 2025. Issuers froze more than \$1.5 billion in total. Yet under a conservative definition, 7.3\% of sanctioned USDT addresses and 18.7\% of sanctioned USDC addresses had already been drained to zero
balances by the time the freeze took effect. 
The empirical evidence further reveals a clear arms race over ordering power: from issuer side out-of-gas execution failures, to public gas bidding, to private transaction submission, and ultimately to evaders delegating execution to specialized MEV infrastructure, making sanction evasion increasingly embedded in the MEV supply chain.

While the empirical evidence establishes the phenomenon, understanding its long-run systemic consequences requires theory. We therefore develop the first game-theoretic model of stablecoin sanctions that explicitly incorporates MEV. The model yields three central results. First, in a reactive sanctions environment, compliant issuers cannot remain outside the ordering market. Second, fixed participation costs push evasion toward a small set of specialized, MEV-aware actors. Third, the implicit MEV tax extracted by block producers grows monotonically with regulatory penalty intensity. Taken together, these results show that as long as ordering power is controlled by profit-driven infrastructure, on-chain sanctions are not a one-off exercise of authority, but an ongoing confrontation with escalating enforcement costs.

The model also shows that repeated interaction creates structural incentives for issuers to vertically integrate into block-building infrastructure, placing direct pressure on blockchain neutrality. More broadly, the same structural vulnerability applies to any privileged action that must be implemented through ordinary transactions, e.g., emergency pauses, governance interventions, or judicial freezes. Whenever ordering power is allocated by economic incentives, a gap between nominal authority and effective enforcement becomes unavoidable.

\subsection{Our Contributions:}

\Paragraph{Problem definition.} We define reactive on-chain sanctions and identify the disconnect between contract-layer authority and ordering-layer enforcement as a system-level security limitation.

\Paragraph{Empirical measurement.} We construct the first sanctions dataset for Ethereum USDT/USDC and quantify the gap between nominal sanctions actions and realized freezing outcomes.

\Paragraph{Mechanism identification.} We introduce the concept of SE-MEV and provide the first empirical evidence that ordering-layer actors can extract MEV from the priority race between sanctioning and evading parties.

\Paragraph{Equilibrium analysis.} We develop a game-theoretic model of sanctions with MEV, showing that the hidden enforcement cost grows without bound as regulatory penalties intensify, and revealing structural incentives for issuers to vertically integrate into block-building infrastructure.

\section{Preliminaries}

\subsection{Blockchain and Smart Contracts}

Blockchains maintain a shared global ledger among mutually untrusted participants. Users submit transactions that request state transitions, and a consensus protocol establishes a canonical transaction order so that all honest nodes deterministically apply the same state updates on-chain. Blockchain platforms differ in their state models. Bitcoin~\cite{nakamoto2008bitcoin} adopts a UTXO-based design, whereas Ethereum uses an account-based model in which persistent state is associated with accounts and updated directly by transactions.
Ethereum further supports smart contracts, enabling general-purpose on-chain computation. Both externally owned accounts (EOAs) and contract accounts coexist in the global state: EOAs initiate transactions, while contract accounts encapsulate executable code and storage. Contract execution is deterministic under the Ethereum Virtual Machine (EVM), and successful executions commit state changes on-chain.

\subsection{Digital Assets and Tokens}

\noindent \textbf{Native Tokens.}
Blockchains natively support protocol-defined digital assets, referred to as native tokens, which are primarily used to pay transaction fees and incentivize consensus participation. Transfers of native tokens are enforced directly by the protocol.

\noindent \textbf{Fungible Tokens.}
Programmable blockchains additionally support application-level fungible tokens implemented via smart contracts. On Ethereum, the ERC-20 standard~\cite{vogelsteller2015eip20} specifies a common interface for such tokens, including balance management, transfers, and delegated transfers through the \texttt{approve} mechanism. 

\noindent \textbf{Stablecoins.}
Stablecoins are fungible tokens designed to track external reference assets, most commonly fiat currencies. They can be classified into decentralized designs, which rely on on-chain collateral and algorithms, and centralized designs, which are backed by off-chain reserves and governed by a centralized issuer. As of February 2026, centralized stablecoins dominate the market, with USDC and USDT jointly exceeding \$255 billion in total market capitalization.

\subsection{Regulation and Sanctions}

Blockchain systems enable permissionless value transfer, which complicates the enforcement of traditional financial regulations. Unlike decentralized protocols, centralized stablecoin issuers are identifiable legal entities and are therefore subject to regulatory mandates, including financial sanctions imposed by authorities such as U.S. The Office of Foreign Assets Control (OFAC)~\cite{ofac_treasury}.

To comply with such mandates, centralized stablecoin contracts (e.g., USDC and USDT) expose issuer-controlled administrative functions that enforce sanctions directly on-chain. By invoking privileged operations such as \texttt{blacklist} or \texttt{freeze}, the issuer updates contract state to disable transfers and approvals involving designated addresses, causing subsequent attempts to move funds to deterministically fail.

As a result, regulatory sanctions are realized as on-chain state transitions initiated by the issuer, which must compete for transaction inclusion and ordering with other transactions. This execution model creates observable races between sanction enforcement and attempts to relocate assets before restrictions take effect.

\subsection{The MEV Supply Chain}
\label{ssec:mev}

MEV refers to the profit that can be extracted by blockchain participants, miners in Proof-of-Work (PoW) or validators in Proof-of-Stake (PoS), by virtue of their power to arbitrarily order, include, or censor transactions within a block.

The public mempool, where pending transactions are broadcast, serves as the arena for MEV extraction. In modern PoS architectures like Ethereum's Proposer-Builder Separation (PBS), this process has evolved into a sophisticated, multi-role supply chain~\cite{yang2025decentralization, zhou2021high, qin2023blockchain, zhou2023sok}:

\begin{itemize}
    \item \textbf{Searchers:} These are independent actors who continuously monitor the mempool for profitable opportunities (e.g., arbitrage, liquidations, or in our context, conflicting sanction/evasion transactions). They express their desired transaction order and execution as "bundles."
    
    \item \textbf{Builders:} These specialized, trusted entities aggregate bundles from multiple searchers and construct the most profitable full block possible. They are responsible for the complex task of optimal block construction.
    
    \item \textbf{Validators (Proposers):} In the PBS model, validators delegate the block construction task. They simply select the block header from the builder that offers them the highest "bribe" or payment, then propose that block to the network.
\end{itemize}

This "bribe" is the technical mechanism by which actors compete for block space and execution priority. This competition manifests in several key forms: (1) a public \textit{priority gas auction (PGA)}, where actors compete using high transaction fees in the public mempool; (2) a private, \textit{bundle-based bidding war}, where searchers submit bundles with an explicit payment directly to builders (e.g., via services like Flashbots~\footnote{https://www.flashbots.net/}); and (3) \textit{direct proposer payments}, where an actor transfers ETH directly to the block producer's fee-recipient address via a standalone on-chain transaction, outside the standard gas-fee or bundle-tip mechanism. True off-chain side payments are unobservable on-chain and therefore out of scope.
A sanctioned entity attempting to relocate funds before enforcement and an issuer attempting to enforce sanctions are thus engaged in this multi-faceted bidding war. The entire MEV supply chain is composed of rational, profit-maximizing actors who will include the transaction or bundle that yields the highest profit, regardless of the underlying semantics (i.e., sanction vs. evasion).

\section{Threat Model}
\label{sec:threat_model}

In this section, we describe the threat model for on-chain sanction enforcement. We specify the scope of systems under consideration, define the strategic roles involved in sanction enforcement and evasion, and clarify the observability limits of our analysis.
\begin{figure*}[htb]
    \centering
    \includegraphics[width=0.8\linewidth]{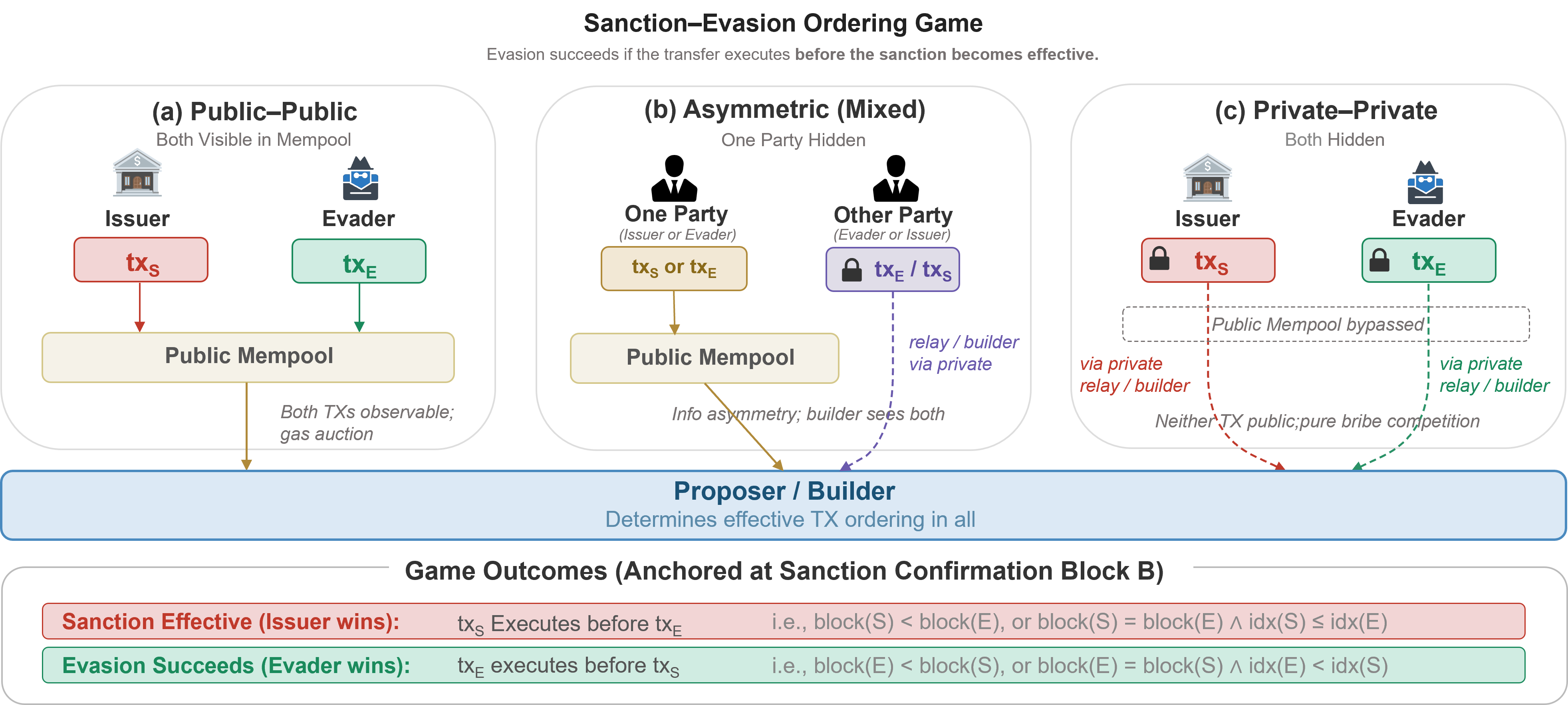}
    \caption{Reactive sanction races escalate into ordering competition across public and private regimes.}
    \label{fig:game}
\end{figure*}

\subsection{Scope of Study}

Our threat model focuses on sanction enforcement in centralized, fiat-collateralized stablecoins, where issuers retain privileged on-chain control. In particular, we consider USDC and USDT, which together account for over 80\% of the global stablecoin market capitalization according to DeFiLlama~\cite{DefillamaStablecoins}. As a result, sanction events involving these assets constitute the dominant form of real-world on-chain stablecoin enforcement.
We focus our analysis on Ethereum, where both issuers maintain their largest circulating supplies, approximately \$48B of USDC and \$99B of USDT as of February 2026~\cite{circle,tether}. Ethereum therefore represents the most significant and adversarially relevant setting for on-chain sanction enforcement and evasion.

\subsection{System Roles}
\label{sec:roles}

We consider three strategic roles in the on-chain sanction process.

\noindent \textbf{Issuer.}
The issuer is the centralized operator of a privileged stablecoin contract responsible for enforcing sanctions. The issuer’s objective is to prevent sanctioned accounts from transferring assets once a sanction decision is made, while minimizing unintended disruption to other users. To this end, the issuer can submit on-chain transactions that update the state of the stablecoin contract, freezing designated addresses and disabling further transfers involving those accounts. In practice, sanction enforcement may be carried out through a single or a sequence of transactions, submitted via either public mempools or private channels.

\noindent \textbf{Evader.}
The evader is a sanctioned entity that controls accounts holding the issuer’s stablecoin balance and seeks to retain access to these assets despite impending enforcement. The evader’s objective is to move funds out of sanction-targeted accounts before the issuer’s on-chain sanction update takes effect. To this end, the evader attempts to submit on-chain transactions that transfer or otherwise relocate assets prior to enforcement, competing for timely inclusion and favorable ordering through either public or private submission channels.
\emph{Note that a sanctioned entity is treated as an evader only when sanction evasion intent can be attributed to it with observable on-chain evidence.} Accordingly, we do not consider shared public infrastructure where individual intent cannot be isolated (e.g., \texttt{Tornado Cash} core contracts), nor large compliant exchange deposit addresses that primarily act as aggregation sinks, as evaders. 

\noindent \textbf{Block Producer.}
The block producer is the entity responsible for constructing blocks, with discretion over transaction inclusion and ordering. Its objective is to maximize economic returns from block production, and it can prioritize, reorder, or exclude transactions when assembling a block. As a result, both issuer and evader transactions are ultimately subject to the block producer’s economically motivated decisions.

Note that all roles are subject to the same network conditions, such as propagation delays and network congestion, and neither issuer nor evader directly controls transaction inclusion or ordering.

Figure~\ref{fig:game} summarizes the resulting ordering game under three visibility regimes. In the simplest case (a), both the sanction transaction and the evasion transaction are broadcast to the public mempool, and the contest reduces to an observable gas auction. When one party routes its transaction through a private channel (b), the game becomes asymmetric: the builder observes both transactions but the public mempool reveals only one. In the limiting case (c), both parties bypass the public mempool entirely, and the contest is resolved purely through private bids to builders and proposers.

\subsection{Observability Boundary}
\label{sec:observability}

Our analysis is grounded in a clear distinction between information that is directly observable on-chain, information that is only partially observable, and aspects that are inherently unobservable within the blockchain system.

\noindent\textbf{Directly observable (on-chain).}
We can directly observe on-chain data, including transactions data and their execution traces, token transfers and event logs, block heights and timestamps. These signals allow us to precisely identify when sanctions become effective on-chain and to measure realized asset movements occurring prior to enforcement in a reproducible manner.

\noindent\textbf{Partially observable (auxiliary).}
Some signals relevant to sanction enforcement are only partially observable and cannot be treated as definitive evidence, yet they provide useful context. Due to the incomplete and vantage-point–dependent nature of mempool access, we cannot precisely observe transaction submission times. Entity identities inferred from third-party address labels may be inaccurate or outdated, and the use of private transaction services is largely invisible on-chain. Accordingly, we treat such signals as supporting evidence, rather than as ground truth.

\noindent\textbf{Unobservable (off-chain).}
Certain aspects of the sanction process are not publicly observable. In particular, we cannot directly observe off-chain decision timelines, such as when an issuer internally decides to initiate a sanction, nor private off-chain arrangements, such as bilateral agreements or side payments between actors. Our analysis does not rely on access to such information and is limited to what can be established from observable evidence.

\section{Methodology: Dataset Construction}
\label{sec:method}

\subsection{Phase~0: Raw On-Chain Data Collection}
\label{sec:phase0}

Phase~0 collects raw on-chain data related to sanction enforcement, consisting of sanctioned addresses and their stablecoin transfer activity. These data form the raw dataset used by all subsequent phases of our pipeline.

\noindent\textbf{Time span and data sources.}
We collect data from 2017-11-28 12:41:21 UTC (the USDT contract creation time) to 2025-08-31 23:59:59 UTC using Google’s public BigQuery Ethereum dataset and public archive nodes.

\begin{figure*}[thb]
    \centering
    \includegraphics[width=0.75\linewidth]{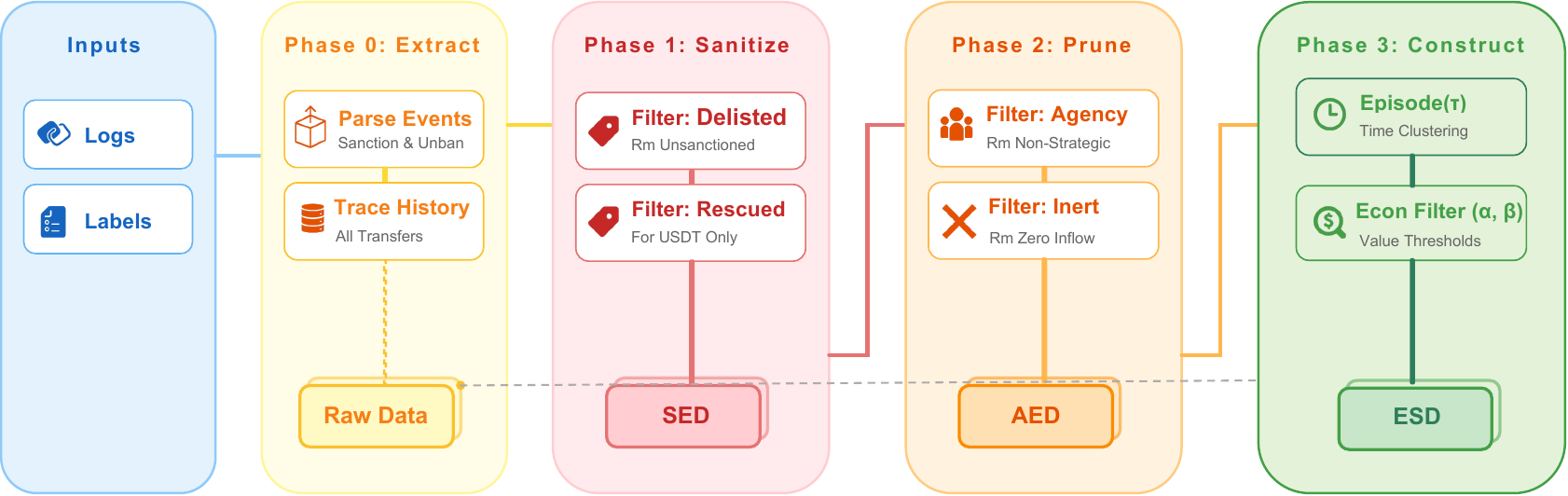}
    \caption{Overview of the dataset-construction pipeline}
    \label{fig:pipeline}
\end{figure*}
\noindent\textbf{Sanction event extraction.}
We scan execution traces and event logs of the USDT and USDC contracts to extract all sanction related operations, including blacklist updates and, for USDT, the \texttt{destroyBlackFunds} operation. For each operation, we record the affected address, token, block number, and transaction identifier.

\noindent\textbf{Raw sanctioned address set.}
From sanction events, we derive a raw set of sanctioned addresses, without applying any filtering.

\noindent\textbf{Stablecoin transfer and approval histories.}
For each sanctioned address, we retrieve the complete on-chain history of the corresponding stablecoin, including all incoming and outgoing transfers, approval events, and reverted transfer attempts. This yields a raw activity dataset capturing all token movements and delegated spending authorizations involving sanctioned addresses over time.

\subsection{Phase~I: Sanction Semantic Filtering}
\label{sec:phase1}

Since on-chain sanction actions do not always reflect true sanction intent, Phase~I applies semantic filtering based on surrounding on-chain activity to remove addresses that are not representative of actual sanction enforcement, yielding a cleaned set of sanctioned entities for downstream analysis.

\noindent\textbf{Revoked sanctions.}
We remove addresses whose sanctions are later lifted through explicit unblacklist operations. 
Such cases typically correspond to temporary testing actions, false positives, or compliance remediation, and do not constitute persistent sanction enforcement suitable for evasion analysis.

\noindent\textbf{Recovery-motivated sanctions.}
For USDT, we remove sanctions that are triggered as part of issuer-driven fund recovery rather than enforcement. In particular, when tokens are accidentally sent to irrecoverable or unintended addresses (e.g., \texttt{0x0}), Tether may blacklist the address to prevent permanent loss, destroy the frozen balance, and subsequently re-issue the funds to a treasury or user-controlled address (e.g., a freeze $\rightarrow$ destroy $\rightarrow$ re-issuance sequence). Such sanctions are administrative in nature and are initiated to restore funds, not to restrict adversarial behavior. We therefore exclude related addresses from the evasion dataset.

After applying the above semantic filters, we retain a set of sanctioned addresses whose blacklist actions are more likely to reflect genuine sanction enforcement. This filtered address set forms the \emph{Sanctioned Entity Dataset (SED)} used in subsequent phases.

\subsection{Phase~II: Adversarial Actor Filtering}
\label{sec:phase2}

While Phase~I isolates sanction actions that are more likely to reflect genuine enforcement, the resulting SED still contains many sanctioned addresses that are not modeled as \emph{evaders} under our threat model (\S\ref{sec:roles}).
Phase~II therefore filters the SED to remove such cases, retaining only sanctioned entities for which evasion intent and strategic action are observable.

\noindent\textbf{Non-Strategic Intermediaries.}
We exclude sanctioned addresses that correspond to shared public infrastructure or compliant aggregation sinks, such as mixer core contracts, generic routing contracts, and large compliant exchange \emph{deposit aggregation clusters}. These addresses may appear in sanctioned fund flows and therefore reflect the outcomes of evasion or enforcement actions. However, they are not the entities that initiate or decide such actions. Because evasion intent cannot be attributed to these endpoints as independent actors, we do not treat them as adversarial evaders in our analysis.

\noindent\textbf{Inert addresses.}
We further exclude sanctioned addresses that never hold any stablecoin balance, i.e., addresses with no observed stablecoin inflow. Such addresses have no assets to relocate and therefore cannot exhibit sanction evasion behavior. 

After applying the above filters, we obtain a refined set of sanctioned entities that both control stablecoin assets and exhibit observable capacity for evasion. This filtered population forms the \emph{Adversarial Evasion Dataset (AED)}, which serves as the primary input for the final phase.

\subsection{Phase~III: Evasion Episode Identification}
\label{sec:phase3}

After identifying adversarial actors, the remaining challenge is to define what constitutes concrete \emph{sanction evasion} behavior. On-chain data is recorded as discrete transactions, whereas evasion corresponds to a sustained behavioral intent to relocate assets prior to sanction enforcement. Labeling individual transfers as evasion would therefore be arbitrary and fragile: evasion may unfold across multiple transactions, be interleaved with routine usage, or occur through contract-mediated execution paths.

Phase~III addresses this ambiguity through \emph{behavioral abstraction}. We reconstruct raw transaction streams of adversarial entities into temporally coherent \emph{intent-level evasion episodes}, each representing a contiguous period of activity plausibly driven by a single evasion attempt. These episodes serve as the fundamental units for identifying and analyzing sanction--evasion behavior.

\noindent\textbf{Intent episode construction.}
For each address in the AED, we first order all of its transactions by block timestamp and compute the time gaps between consecutive transactions. We then aggregate these inter-transaction gaps across all addresses and estimate their distribution using kernel density estimation (KDE). 
We define a threshold $\tau$ at the first prominent valley of this distribution, which corresponds to the transition between short gaps associated with temporally related actions and longer gaps indicative of unrelated activity. Using $\tau$, we partition each address’s transaction history into disjoint, time-ordered \emph{intent episodes}, where consecutive transactions within an episode are separated by less than $\tau$.

\noindent\textbf{Episode-level economic filtering.}
Intent episodes vary widely in economic significance. To distinguish plausible sanction evasion from routine usage or dust-level activity, we jointly quantify each episode’s economic scale and apply a conservative materiality criterion grounded in regulatory practice and economic rationality.

For an intent episode $S$, We define the total available liquidity during an episode as
$L_{\text{episode}} = B_{\text{start}} + \sum_{\text{tx} \in \text{episode}} \text{Inflow}(tx)$,
where $B_{\text{start}}$ is the stablecoin balance immediately before the episode begins, and $\text{Inflow}(tx)$ denotes the stablecoin amount received by the address in transaction $tx$ during the episode. This quantity captures the maximum amount of stablecoin the entity could plausibly relocate within the episode.

Let $V_{\text{out}}$ denote the total stablecoin outflow within the episode. We summarize the episode’s asset relocation using the outflow ratio $V_{\text{out}} / L_{\text{episode}}$.
We classify $S$ as an evasion episode only if it satisfies
$(V_{\text{out}} \ge \beta)\ \land\ \left( \frac{V_{\text{out}}}{L_{\text{episode}}} \ge \alpha \right)$.
The absolute threshold $\beta$ is set to \$1{,}000, reflecting widely adopted international AML standards (e.g., the FATF Travel Rule). The relative threshold $\alpha$ captures the fraction of available liquidity relocated. We evaluate the robustness of our results under different choices of $\alpha$ in later analysis.

Applying this procedure yields an \emph{evasion episodes} dataset, each representing a temporally coherent and economically material attempt to relocate assets in anticipation of sanction enforcement. This dataset forms the basis for our empirical analysis of sanction--evasion races.

\subsection{Auxiliary Data Collection}

In addition to the datasets produced by the above pipeline, we collect a set of auxiliary data that are used where needed throughout the pipeline and in subsequent empirical analyses.

\noindent\textbf{Governance forensics.}
Because sanction enforcement is realized through issuer-controlled privileged functions, execution logs alone may miss failed or reverted enforcement attempts. To capture the full set of issuer-initiated sanction activity, we perform targeted governance forensics to reconstruct sanction-related actions initiated by authorized controllers of USDT and USDC. For USDT, this involves tracing blacklist operations initiated by multisig owners. For USDC, which employs an upgradeable proxy architecture, we follow role changes over time and track blacklist actions issued by authorized roles \texttt{blacklister}. This reconstruction allows us to observe sanction \emph{attempts} rather than only successful executions, providing a more complete view of issuer enforcement behavior.

\Paragraph{USDT multisig timestamp reconstruction.}
For USDT, we further distinguish between the commitment time
and the effective time in the multisignature enforcement workflow.
Specifically, we record the timestamp at which sanction intent is first
confirmed on-chain via \texttt{Submission}, denoted
$t_{\text{submit}}$, and the later timestamp at which the blacklist
actually becomes effective via \texttt{Execution}, denoted
$t_{\text{exec}}$.
Throughout the paper, our main sanction-aware evasion labels,
the $\Delta$ statistics, and the race / tactical-reactive / strategic-migration
regimes remain uniformly anchored at $t_{\text{exec}}$.
By contrast, $t_{\text{submit}}$ is used only as a mechanism-level
interpretive variable to characterize an issuer-side
committed-but-not-yet-effective exposure window, rather than to
redefine evasion itself.

\noindent\textbf{Address attribution.}
To support actor classification and filtering throughout the pipeline, we collect address labels from multiple sources, including Etherscan, Dune labels, and third-party security intelligence providers. We adopt a conservative, hierarchical fusion strategy: when labels are consistent, we retain the attribution; when labels conflict or supporting evidence is insufficient, we assign the address as \emph{unknown}. This approach avoids forced attribution and limits the introduction of systematic bias, while providing auxiliary signals for identifying issuers, intermediaries, and sanctioned entities where attribution is reliable.

\noindent\textbf{Regulatory and legal anchors.}
To contextualize on-chain sanction activity, we collect external regulatory records, including the OFAC SDN list and publicly available court seizure warrants indexed via the CourtListener API. These records serve as temporal and semantic anchors linking on-chain enforcement actions to documented regulatory or legal events. We do not assume that such off-chain records are complete, timely, or perfectly aligned with on-chain execution; rather, they are used as reference points to support interpretation of sanction timing and issuer behavior.

\noindent\textbf{Transaction submission timing.}
To estimate transaction submission timing prior to confirmation, we collect public mempool observations as a conservative lower bound. We aggregate two mempool archives, Blocknative ``Mempool Archive''  (2019-01-11 to 2025-01-03) and the Flashbots ``Mempool Dumpster'' (2023-07-08 to 2025-08-31), and record, for each transaction, the earliest observed \texttt{first-seen} timestamp. 

\noindent\textbf{Private submission signals.}
To identify transactions that may bypass the public mempool, we collect auxiliary signals related to private submission. This includes private-transaction flags from mempool sources and Dune data (Flashbots Protect and MEV-Share).

\noindent\textbf{Fee Baselines.}
To isolate sanction--evasion bidding behavior from ambient network congestion, we compute block-level fee baselines for comparison. For each block that contains a blacklist transaction or a candidate evasion episode, we calculate the median of effective transaction fees in that block. These baselines characterize the contemporaneous fee environment faced by all transactions in the same block, allowing us to assess whether issuer or evader transactions pay fees that are unusually high relative to background conditions, while remaining robust to outliers.

\subsection{Limitations}

Our measurements should be interpreted as \emph{conservative lower bounds}. Several sources of incompleteness and limited observability constrain what can be inferred from the available data.

\noindent\textbf{Mempool observability.}
Mempool visibility is vantage-point dependent and inherently incomplete. First-seen timestamps provide only lower bounds on transaction submission time, and historical coverage is limited to available archives.

\noindent\textbf{Address attribution.}
Address labels may be noisy or stale despite cross-source validation. 

\noindent\textbf{Regulatory records.}
Judicial and regulatory data may be delayed, partial, or unavailable (e.g., sealed cases), and non-U.S.\ sanctions are outside our primary scope.

\noindent\textbf{Private transaction channels.}
Coverage of private routing (e.g., Flashbots) is incomplete. Private-transaction signals are treated as indicative evidence, not exhaustive measurements.

\noindent\textbf{Off-chain coordination.}
We observe standard transaction fees and direct proposer payments (standalone ETH transfers to fee-recipient addresses) on-chain.  Purely off-chain side payments are unobservable and not quantified.

Overall, these limitations bias our results toward underestimating the prevalence, intensity, and sophistication of sanction--evasion behavior rather than overstating it.

\section{Empirical Study of On-chain Sanction}
\label{sec:empirical}
\subsection{The Landscape of On-Chain Sanctions}
\label{sec:landscape}

We begin by characterizing the empirical landscape of on-chain sanctions after applying the dataset construction
pipeline described in Section~\ref{sec:method}. From raw execution logs, we extract
2{,}418 sanctioned USDT addresses and 341 sanctioned USDC addresses (Phase~0).
After applying Sanction Semantic Filtering (Phase~I) and Adversarial Actor Filtering (Phase~II),
we obtain the final \emph{Adversarial Evasion Dataset (AED)}, comprising sanctioned addresses
that remain consistent with our threat model and retain observable evasion potential.
The AED contains 1{,}990 USDT addresses and 123 USDC addresses. In addition, we identify \$24{,}808{,}527 of USDT that are destroyed via \texttt{destroyBlackFunds}
as part of issuer-driven recovery workflows.
Unless otherwise stated, the analysis in this subsection is address-level, i.e., each observation corresponds to one sanctioned address in the AED.

\begin{figure}[hbpt]
    \centering
    \includegraphics[width=1\linewidth]{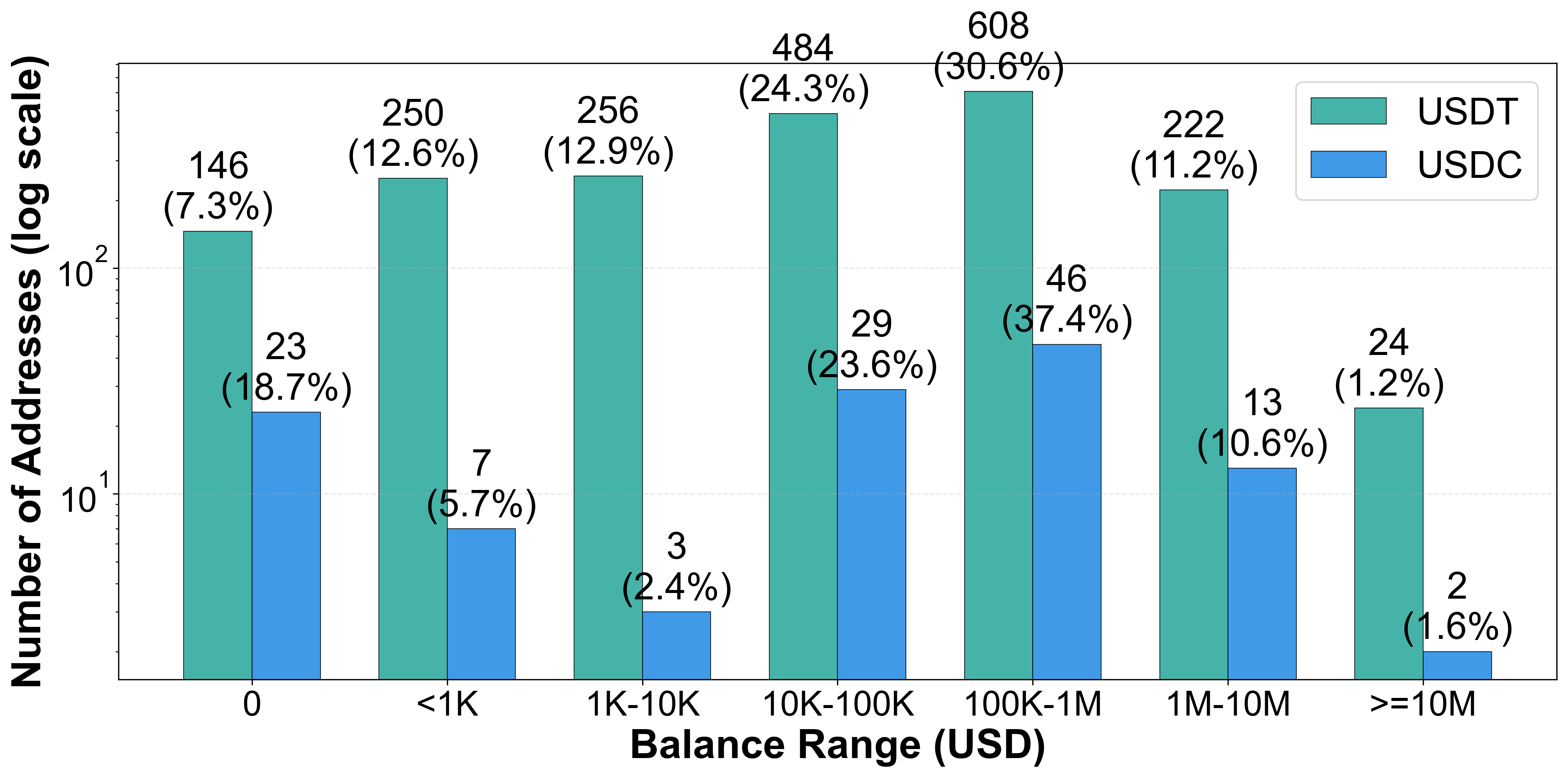}
    \caption{Statistics for Frozen Values}
    \label{fig:frozen_value}
\end{figure}

\begin{table}[tbhp]
\centering
\caption{Filtering of sanctioned addresses into AED. Percentages are relative to the Phase~0 raw set.}
\label{tab:addr_filtering}
\footnotesize
\setlength{\tabcolsep}{4pt}
\begin{tabular}{@{}lcc@{}}
\toprule
\textbf{Phase} & \textbf{\# USDT addr.} & \textbf{\# USDC addr.} \\
\midrule
P0: Raw              & 2418 (100.0\%) & 341 (100.0\%) \\
P1: Filtered         & 352 (14.6\%)   & 102 (29.9\%)  \\
P1: Output (SED)     & 2066 (85.4\%)  & 239 (70.1\%)  \\
P2: Filtered         & 76 (3.1\%)     & 116 (34.0\%)  \\
P2: Output (AED)     & 1990 (82.3\%)  & 123 (36.1\%)  \\
\bottomrule
\end{tabular}
\end{table}

We next characterize the \emph{provenance} of sanctioned addresses in the AED.
Cross-referencing the dataset with the OFAC Specially Designated Nationals (SDN) 
list shows that
only a small fraction of sanctioned addresses can be directly linked to publicly listed OFAC designations,
specifically 26 USDT addresses and 15 USDC addresses.
To further contextualize the remaining cases, we query publicly available litigation records via CourtListener and identify 121 additional USDT addresses and 18 USDC addresses associated with documented legal proceedings.
Beyond publicly disclosed regulatory designations and litigation records, we observe a distinct behavioral pattern among a subset of sanctioned addresses: repeated transfers of zero-value amounts to other addresses.
Such activity is widely recognized as a phishing or address-poisoning technique used to create confusion in transaction histories and lure victims into misdirected transfers.
Using this behavioral signature, we further label 431 USDT addresses and 30 USDC addresses as phishing-like.
The remaining sanctioned addresses cannot be directly attributed to public regulatory disclosures, litigation records, or phishing-like behavior, and are therefore categorized as \emph{unknown}. 

Figure~\ref{fig:sanction_provenance} summarizes the distribution of sanctioned addresses in these categories.

\begin{finding}
In aggregate, over 94\% of on-chain sanctions are initiated through issuer-driven enforcement or compliance actions without a corresponding public OFAC designation.
\end{finding}

From this point onward, we keep the sanctioned-address population fixed and switch only the measurement axis from address counts to frozen value at sanction effectiveness.

We next quantify the \emph{economic impact} of on-chain sanctions by measuring how much stablecoin value is actually frozen at the time of enforcement.
For each address in the AED, we compute its stablecoin balance immediately before the corresponding blacklist transaction becomes effective on-chain.
Aggregating across all sanctioned addresses yields the total amount of value directly frozen by issuer enforcement actions.
Figure~\ref{fig:frozen_value} reports the distribution of frozen balances across sanctioned addresses.
In aggregate, USDT sanctions freeze \$1{,}411{,}842{,}932.22, while USDC sanctions freeze \$108{,}250{,}166.32.
These figures represent the lower bound of enforceable value, as they only account for balances remaining at the moment sanctions take effect.

\begin{figure}[bhtp]
    \centering
    \includegraphics[width=0.8\linewidth]{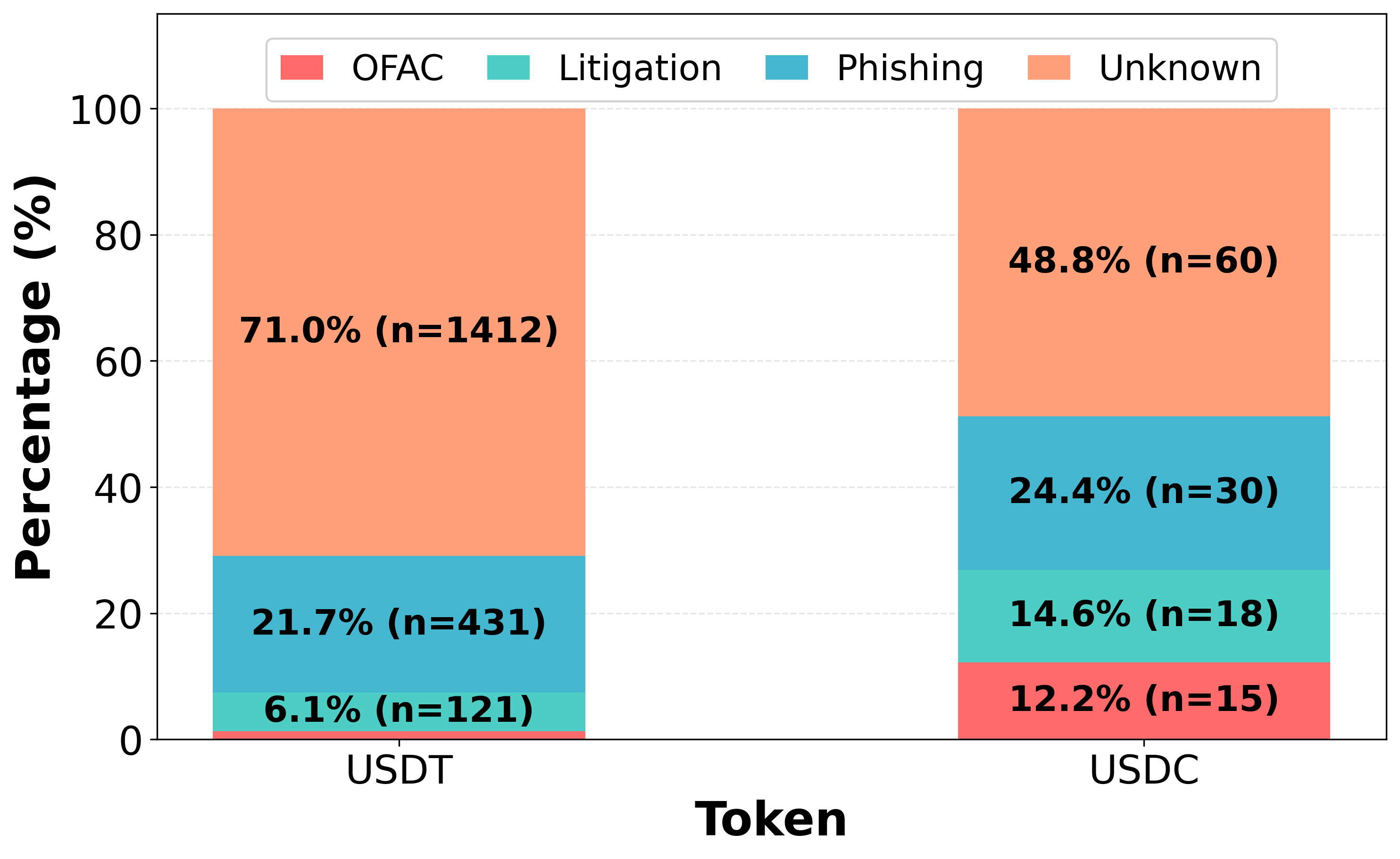}
    \caption{Provenance of Sanctioned Addresses}
    \label{fig:sanction_provenance}
\end{figure}

A closer inspection, however, reveals a notable anomaly.
A non-trivial fraction of sanctioned addresses hold little or no balance at the time of enforcement.
For USDT, 146 addresses (7.34\%) have a zero balance, and an additional 250 addresses (12.56\%) hold less than \$1{,}000.
For USDC, 23 addresses (18.7\%) have a zero balance, and 7 addresses (5.69\%) hold less than \$1{,}000.
Since all addresses in the AED are confirmed to have received stablecoin inflows prior to sanctioning, the presence of zero-balance and low-balance sanctions suggests potential existence of targeted on-chain sanction evasion.

\begin{finding}
Despite over \$1.5B in stablecoins being frozen in aggregate, a substantial fraction of sanctioned addresses hold zero or near-zero balances at the time enforcement takes effect.
\end{finding}

\subsection{Characterizing Evasion Behavior}
\label{sec:evasion_behavior}

The previous section highlighted a distinct anomaly: certain addresses depleted their balances shortly before on-chain sanctions took effect. This phenomenon requires a behavioral explanation. Focusing on AED, we investigate three complementary dimensions: \textit{when} pre-sanction capital flight occurs, \textit{where} the funds flow, and \textit{how} each outflow is executed.

\Paragraph{When: Temporal Structure of Pre-Sanction Outflows.}

To answer ``when'', we measure pre-sanction outflows relative to
$t_{\text{exec}}$, the on-chain time at which blacklist enforcement becomes
effective. A single transaction is too granular as an analysis unit; evasion may span multiple transfers or interleave with daily activities.
Therefore, we aggregate raw outflow transactions into intent-level \textit{evasion episodes} and characterize them by their temporal distance relative to the sanction enforcement.

For each AED, we sort its transactions by $t_{\text{exec}}$ and compute the time intervals between adjacent transactions. By aggregating the interval distributions across all addresses and applying Kernel Density Estimation (KDE), we extract a data-driven time threshold $\tau = 107$ seconds based on the first local minimum of the density curve. Adjacent transactions with intervals smaller than $\tau$ are merged into a single intent episode. To filter out background noise, we impose economic materiality constraints: we require the episode's net outflow volume $V_{\text{out}} \geq \$1{,}000$ and its ratio to the address's total liquid assets 
$V_{\text{out}}/L_{\text{episode}} \geq 10\%$. This yields 18,432  economically meaningful episodes. 
We further evaluate the robustness of this materiality criterion by varying the relative outflow threshold $\alpha$ and report the stability results in Appendix~\ref{app:alpha-robustness}.

The time difference $\Delta$ between the final net outflow of an episode and the sanction enforcement spans from seconds to years. The distribution of $\log(\Delta)$ exhibits significant heavy-tailed and multi-modal characteristics, indicating profound temporal heterogeneity in evasion behaviors. We fit a Gaussian Mixture Model (GMM) to $\log(\Delta)$ and determine the optimal number of components $k=30$ using the Bayesian Information Criterion (BIC) over $k \in [1,50]$. Treating this GMM as a density estimator, we merge adjacent components along low-density valleys to obtain objective cutoff points.

As summarized in Table~\ref{tab:regime}, this yields four regimes: (1) \emph{Race}, $\Delta \leq 242$ seconds; (2) \emph{Tactical reactive}, $242 \text{ s} < \Delta \leq 26.5 \text{ hours}$; (3) \emph{Strategic migration}, $26.5 \text{ hours} < \Delta \leq 88 \text{ days}$; and (4) \emph{Long-tail}, $\Delta > 88 \text{ days}$. We exclude the long-tail regime from subsequent sanction-aware analysis due to its weak causal coupling with the sanction event.

Three phenomena stand out. First, the Race regime is exceedingly rare: we identify only 9 episodes for USDT and none for USDC. The absence of
USDC races is unsurprising given its smaller sanction footprint in our dataset, but the overall scarcity of race-window episodes suggests that
most evasion does not occur as last-moment ordering contests; rather,
funds typically move earlier, when targets have already perceived sanction risk. Second, despite their rarity, Race episodes are high-value events: their median outflow is $\$3.83\text{M}$, about $50\times$ larger than the $\$64\text{K}$--$\$70\text{K}$ medians in the tactical-reactive and strategic-migration regimes. This indicates that last-moment ordering contests mainly arise when substantial balances are at stake, whereas
earlier evasion regimes consist of more diffuse and lower-value flows.
Third, strategic-migration serves as the primary conduit for capital flight: in USDT alone, it comprises 6,232 episodes across 517 unique addresses, accounting for $\$1.5\text{B}$ in pre-sanction outflows, dwarfing the tactical-reactive regime in magnitude.

\begin{table*}[htpb]
\centering
\caption{Evasion episode statistics by temporal regime and 
token. Regime boundaries are derived from the GMM-based 
segmentation of $\log(\Delta)$. Human-readable durations: 
$242\,\text{s} \approx 4\,\text{min}$, 
$95{,}514\,\text{s} \approx 26.5\,\text{h}$, 
$7{,}614{,}341\,\text{s} \approx 88\,\text{d}$. 
Outflow values are in USD.}
\label{tab:regime}
\small
\begin{tabular}{llrrrrrr}
\toprule
\textbf{Regime} & \textbf{Token} & \textbf{\# Ep.} 
  & \textbf{\# Addr.} & \textbf{Total Outflow (\$)} 
  & \textbf{Median Outflow (\$)} 
  & \textbf{Median $\Delta$ (s)} 
  & \textbf{Min $\Delta$ (s)} \\
\midrule
\multirow{2}{*}{Race ($\Delta \leq 242\,\text{s}$)}
  & USDT & 9 & 9 & 25,721,865 & 3,827,942 & 72 & 24 \\
  & USDC & 0 & 0 & --- & --- & --- & --- \\
\midrule
\multirow{2}{*}{\shortstack[l]{Tactical reactive\\($242\,\text{s} < \Delta \leq 26.5\,\text{h}$)}}
  & USDT & 494 & 222 & 150,867,282 & 66,514 & 38,751 & 276 \\
  & USDC & 10 & 7 & 30,246,015 & 1,054,288 & 39,852 & 17,256 \\
\midrule
\multirow{2}{*}{\shortstack[l]{Strategic migration\\($26.5\,\text{h} < \Delta \leq 88\,\text{d}$)}}
  & USDT & 6,232 & 517 & 1,551,578,583 & 69,879 & 2,431,194 & 95,688 \\
  & USDC & 99 & 37 & 76,111,317 & 33,066 & 2,529,168 & 114,000 \\
\midrule
\multirow{2}{*}{Long-tail ($\Delta > 88\,\text{d}$)}
  & USDT & 11,489 & 351 & 2,015,455,042 & 64,000 & 35,939,484 & 7,616,620 \\
  & USDC & 99 & 29 & 13,123,274 & 30,535 & 43,469,848 & 7,856,556 \\
\bottomrule
\end{tabular}
\end{table*}

\begin{finding}
Evasion fund transfers unfold across distinct temporal scales: short-term
``Race'' episodes are rare but carry a substantial monetary premium,
whereas multi-day to multi-month ``Strategic migration'' accounts for the
vast majority of pre-sanction capital flight.
\end{finding}

\Paragraph{Where: Destinations of Evaded Funds.}
We next trace the destinations of the evaded capital. Aggregating the receivers of all AED outgoing transfers yields 6,419 unique receiver addresses, absorbing a total of $\$4.05\text{B}$. Using multi-source labels, we classify these destinations into four functional buckets:
compliant CEXs, DeFi infrastructure (DEXs/Aggregators/Bridges),
Obfuscation (Mixers \& No-KYC Swaps), and High-Risk/Sanctioned
Exchanges, with \textit{Others} covering unlabeled or non-attributable
destinations outside these categories (see Table~\ref{tab:destinations}).

\begin{table}[t]
\centering
\caption{Destinations of evaded funds.}
\resizebox{1\linewidth}{!}{
\label{tab:destinations}
\footnotesize
\setlength{\tabcolsep}{6pt}
\renewcommand{\arraystretch}{1.1}
\begin{tabular}{@{}lrrrr@{}}
\toprule
\textbf{Bucket} & \textbf{\#Addr} & \textbf{\%addr} & \textbf{Volume} & \textbf{\%vol}  \\
\midrule
Compliant CEX                    & $1{,}137$ & $17.71$ & $\$309.4$M     & $7.63$ \\
DeFi (DEX / Aggregator / Bridge) & $56$      & $0.87$  & $\$112.6$M     & $2.78$  \\
Obfuscation (Mixers and No-KYC swaps)           & $51$      & $0.79$  & $\$40.2$M      & $0.99$   \\
High-Risk / Sanctioned Exchange  & $120$     & $1.87$  & $\$15.1$M      & $0.37$   \\
Other                & $5{,}055$ & $78.75$ & $\$3{,}572.8$M & $88.23$ \\
\midrule
\textbf{Total}                   & $\mathbf{6{,}419}$ & $\mathbf{100.0}$ & $\mathbf{\$4{,}050.1}$M & $\mathbf{100.0}$ \\
\bottomrule
\end{tabular}
}
\end{table}

Among attributable entities, compliant centralized off-ramps dominate. Compliant CEX deposit addresses constitute 17.7\% of the unique receivers, absorbing $\$309\text{M}$ (with Binance, Crypto.com, and Coinbase as the top three counterparties). The use of DeFi as an off-ramp is an order of magnitude smaller: 56 addresses and $\$112.6\text{M}$, primarily distributed across Tokenlon, 1inch, and Uniswap. Notably, obfuscation tools are used in a highly polarized manner. Mixers and No-KYC swaps collectively account for only $\$40.2\text{M}$ (0.99\% of the volume). Furthermore, this volume is not evenly dispersed but heavily concentrated in a few services: MaskEx ($\$35.2\text{M}$ via 16 addresses) and Tornado Cash ($\$0.24\text{M}$ via 1 address).

\begin{finding}
Evaders heavily prioritize liquidity over anonymity. Compliant CEXs are the primary destinations, whereas privacy-enhancing obfuscation services handle less than 1\% of the volume and are highly concentrated among a few providers.
\end{finding}

\Paragraph{How: Execution Mechanisms of Evasion.}
Finally, we investigate the true initiators of the outgoing transfers. After excluding the outflows from  smart contracts, we find that 98.6\% of evasion transactions are directly self-signed by the sanctioned EOAs. The remaining 1.4\% (276 transactions) reveal a delegated execution pattern driven by external callers.

Among these delegated executions, 33 unique external callers each initiated at least two transfers on behalf of sanctioned addresses. In 274 of the 276 transactions, the sanctioned entity had previously
approved a mainstream DeFi infrastructure contract as the spender, and the evasive transfer was later executed through that approved contract.
This DeFi-based delegated execution is highly concentrated, with the top four infra (1inch v5, Permit2, Tokenlon, and MetaMask Swap) accounting for $\sim$90\% of such cases.
Only 2 transactions ($\$173.9\text{K}$) used EOA-spender approvals, where the sanctioned entity had previously approved an EOA as the spender, which then invoked \texttt{transferFrom} directly to move the funds.

\begin{finding}
A delegated execution pattern exists in pre-sanction evasion, and such coordinated capital extractions rely almost entirely ($>$99\%) on leveraging pre-existing permissions granted to mainstream DeFi infrastructures.
\end{finding}

\subsection{SE-MEV: Arms Races over Ordering Power}
\label{sec:semev}

A contract-level \texttt{freeze} right gives issuers authority over sanctioned funds, but not priority in the execution pipeline. Because freeze transactions and evading transfers compete for the same ordering resources, enforcement becomes an ordering contest rather than a purely contract-level action. We define \emph{Sanction-Evasion MEV (SE-MEV)} as the extractable value generated by this contest, and study how ordering power is priced, contested, and captured across increasingly privileged execution channels.

\noindent \textbf{Level 1: Gas War.}
At the first level, freeze and evasion transactions compete in the same public ordering channel, as shown in Figure~\ref{fig:game}(a). Execution priority is purchased through transaction fees, and enforcement can fail either when an evader outbids the issuer or when the issuer's freeze transaction is mispriced or under-gassed, causing it to revert or arrive too late. By examining issuer sanction transactions, we find that eight USDT freeze transactions reverted due to out-of-gas failures between October 2021 and February 2024. These failures represent a basic form of Level-1 breakdown: before enforcement can win an ordering race, the freeze transaction must first be correctly priced and sufficiently provisioned to execute successfully. \footnote{USDT's multisignature workflow introduces a commit-before-execute window, which enlarges the opportunity for evasion. This window covers 116 material episodes totaling \$73.7M, with a median delay of 5.3 hours.}

\begin{figure}[hbpt]
    \centering
    \includegraphics[width=0.8\linewidth]{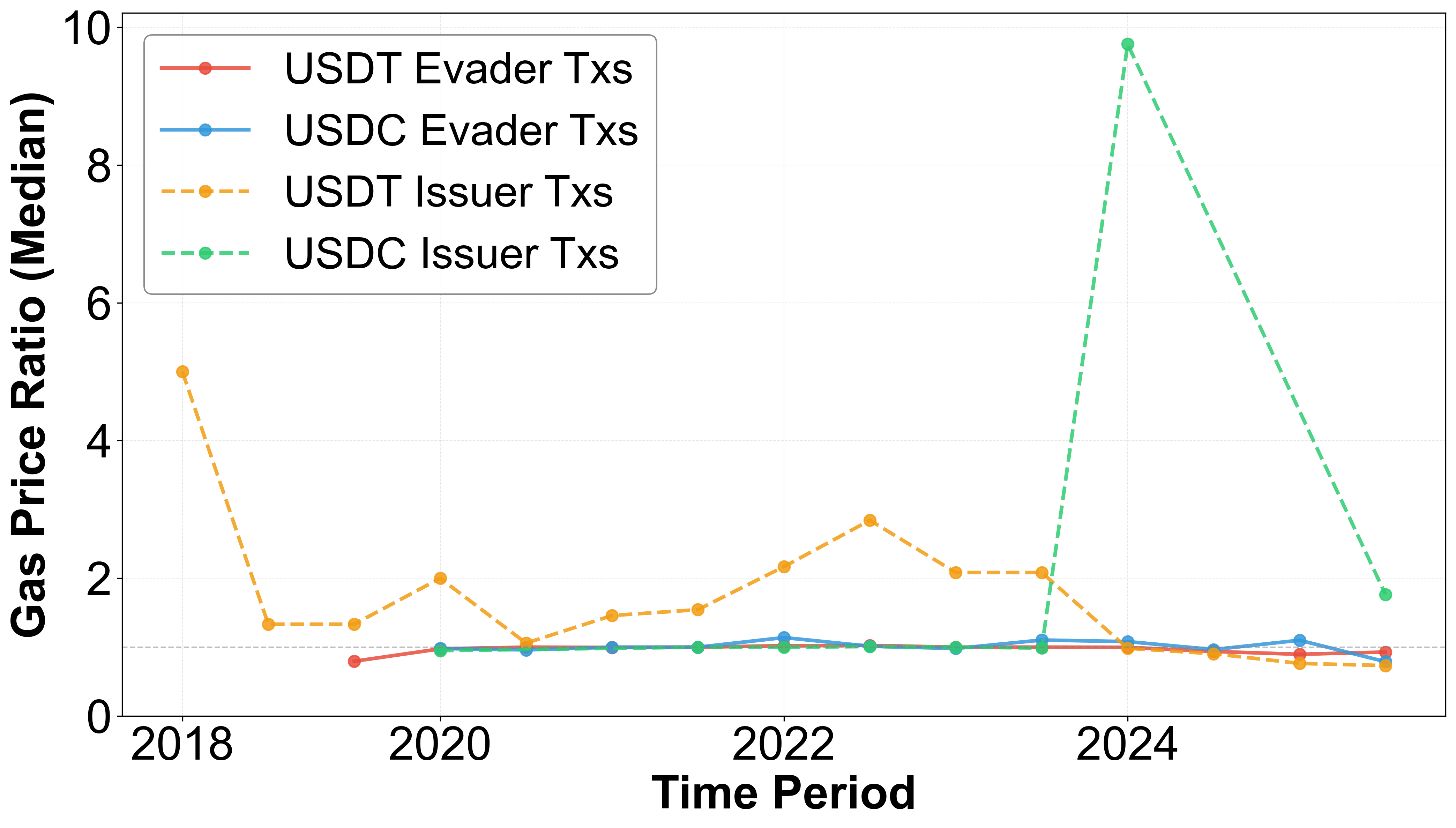}
    \caption{Trend of Gas Price Ratio of Sanction and Evasion}
    \label{fig:gas_war}
\end{figure}

We measure each sanction and evasion transaction's fee aggressiveness by dividing its gas price by the median gas price of the block in which it was included, thereby controlling for background network congestion. Fig.~\ref{fig:gas_war} shows a two-track pattern: USDC issuers bid far above the median throughout the sample, while USDT issuers cluster near or below it early on, then rise sharply from H2 2023 and converge toward the evader-side premium level.

We focus on the race and tactical-reactive windows ($\Delta \leq 26.5$h; 513 episodes, 657 txs), where the ordering competition is most competitive. In short races ($\Delta \leq 1200$s), fee premium increases with episode stake $L_{\mathrm{episode}}$ ($\rho=0.49$, $p=0.0017$), indicating that higher-value episodes induce more aggressive gas bidding. This correlation weakens over the full window, suggesting that longer-horizon evasion relies less on visible fee competition and more on pre-arranged execution paths. We therefore treat these fee-market patterns as evidence consistent with the SE-MEV mechanism, rather than as a stand-alone causal identification (Tab.~\ref{tab:fee-stake}; App.~\ref{app:fee-stake}).

Block-level fee premia capture the final inclusion price, but not the bidding process that occurs before inclusion. Using historical mempool data, we recover five replace-by-fee (RBF) chains, where transactions with the same nonce are repeatedly rebroadcast with higher fees: two issuer-side chains (\(\times 7.5\) within 10\,min in Jul.\ 2020; \(\times 1.36\) over \(\sim\)95\,min in Feb.\ 2021) and three evader-side chains (\(\times 1.17\)--\(\times 2.49\), median \(\times 1.56\); 40\,s--24\,min).

We also observe ordering inversion despite earlier mempool entry. For sanctioned address \texttt{0xe86efa}, the issuer's legacy freeze transaction (\texttt{0xb8521b...6b49}, 50 gwei) entered the mempool 22 min before the evader's EIP-1559 transaction (\texttt{0xff73b7...3d63}, priority fee 179.51 gwei). The evader transaction was included within 10\,s in block 18536174 and moved 100K USDT, while the issuer's legacy transaction, unable to adjust to the base fee, stayed pending for 3.18\,h. This is a representative SE-MEV case in which earlier visibility did not translate into execution priority.

\begin{finding}
     Level 1 of SE-MEV is a fee-based ordering race. Both issuers and evaders pay measurable fee premiums; in short-horizon races, larger stakes induce more aggressive gas bidding, while mispriced or under-gassed freeze transactions can fail before enforcement takes effect.
\end{finding}

Moreover, a substantial fraction of race-window key transactions have no reconstructible mempool record, consistent with private submission. The competition is moving into an information stage.

\Paragraph{Level 2: Visibility Asymmetry.}
Level 2 concerns the choice of submission channel. Private channels, such as Flashbots, MEV-Share, and private RPCs, let transactions bypass mempool-wide visibility and route directly to builders. For issuers, private submission reduces leakage of freeze intent before enforcement takes effect; for evaders, it hides the evading transfer from issuer-side monitoring until inclusion. The contest therefore shifts from paying for priority in the visible fee market to controlling who can observe the transaction before execution (Fig.~\ref{fig:game}(b,c)).

USDT evaders adopt private channels from Oct.\ 2023 (116 cases), while USDT issuers follow only in Jan.\ 2025, about 15 months later, and surge in 2025-Q3 (74 cases). USDC issuers begin in Dec.\ 2024 at a smaller, parallel scale. This timing gap suggests an asymmetric perception of visibility risk: evaders internalize the cost of mempool exposure earlier, whereas issuers migrate only after sustained losses in visible ordering competition.

Channel migration, however, does not eliminate ordering costs. We normalize evader private and non-private transactions by the median gas price of their inclusion block and compare their effective fee premiums. The medians are nearly identical (1.04$\times$ vs.\ 0.99$\times$), but the upper tails diverge sharply: P75 reaches 1.55$\times$ for private transactions versus 1.09$\times$ for non-private transactions, and P90 reaches 3.07$\times$ versus 1.34$\times$. Thus, private submission does not remove the ordering tax; it shifts the payment surface from visible fee competition toward builder-side ordering payments. Full quantile and tail-threshold statistics are reported in App.~\ref{app:private_tx}. Decomposing evader private flow further, we find that only $\sim$50\% of private evader transactions are direct USDT \texttt{transfer} calls; the rest are routed through DeFi resolvers, primarily 1inch.

\begin{finding}
Level 2 of SE-MEV takes the form of submission-channel migration. Both sides move toward private submission. However, private channel do not reduce the tax extracted by infrastructure.

\end{finding}

\Paragraph{Level 3: Delegation to Specialized MEV Infrastructure.}
Level 3 turns ordering competition into oblivious delegation. Evaders wrap evasion as ordinary DeFi intent and let resolvers execute it: cost falls for the evader, while infrastructure, optimizing execution quality and MEV capture in a sanction-agnostic way, is effectively weaponized when its objective happens to align with evasion. Co-option here is structural, not intentional: a system- rather than actor-level phenomenon.

This level appears only evader-side. We identify 27 evasion txs routed via resolvers (e.g., \emph{1inch Rizzolver}, \emph{Seawise: Resolve}, \emph{1inch Fusion Resolver}) in the race/tactical private dataset; back-searching the same resolvers across all race/tactical txs (public and private) yields no additional matches. These txs mostly convert freeze-risk stablecoins into decentralized stablecoins or ETH.

Expanding chain-wide, we find 131 resolver-based cash-outs from AED addresses: the 27 race/tactical txs, 20 long-tail cases, and 84 in the strategic-migration window (median gap 57.15,h). Nine transfer ETH directly to the proposer's fee-recipient address, including one builder transfer exceeding \$100.

Under MEV-Boost, direct transfers and gas payments are functionally equivalent but expose a new value-capture structure: once resolvers enter, the sanction-induced ordering premium is no longer captured solely by builders or proposers, but split across a multi-layer supply chain including resolvers and third-party bots.

\begin{finding}
    
Level 3 of SE-MEV is ordering agency: evaders outsource bidding to specialized third parties, expanding the MEV supply chain from builder–proposer to multi-layer, with tax further redistributed to intermediaries.

\end{finding}

\Paragraph{Summary.}
The three levels coexist and progressively embed value deeper into block-building infrastructure: from no contest, to public premia, to private builder bids, to delegated bidding via intermediaries. Submission form changes; the core fact does not: ordering execution remains the bottleneck between sanction authority and sanction effect, and the MEV supply chain remains the residual claimant. Whichever side wins a given confrontation, value has already flowed into ordering infrastructure. The next section formalizes this mechanism and its implications for protocol neutrality and enforcement.

\section{A Game-Theoretic Model of SE-MEV}
\label{sec:model}

\subsection{Informal Overview}

We model fiat-backed stablecoin sanctions as a \emph{Tullock contest} over transaction ordering~\cite{chowdhury2011generalized, ewerhart2015mixed, eyal2015miner}.
A compliant issuer $I$ and an evader $E$ compete over a fixed illicit balance $V>0$ at a risky address; MEV infrastructure and block proposers determine whether $V$ is frozen or successfully evaded.

A contest is triggered by a publicly observable \emph{off-chain regulatory or risk signal} (e.g., sanctions list updates, public attribution of hacks/fraud, or on-chain risk flags).
We treat this signal as the common start time at which both parties can react.

We study a single \emph{sanction contest} tied to one block.
Issuer $I$ prepares a freezing transaction $Tx_I$ and evader $E$ prepares an evasion transaction $Tx_E$ that attempts to move $V$ to a fresh address.
Both can route transactions via the public mempool (gas price/tip bidding) and private MEV channels (builders, relays, private RPCs).
We abstract these mechanisms into scalar \emph{contest expenditures} $b_I$ and $b_E$, interpreted as the incremental payments and priority fees attributable to the corresponding execution path.

The contest \emph{starts} when both players can simultaneously choose expenditures, and \emph{resolves} in the first block in which a proposer sees mutually exclusive valid bundles corresponding to the freeze and evade paths.
We treat this first-block resolution as a one-shot stage game; subsequent blocks correspond to separate contests.

The terminal outcomes are: (i) \emph{freeze}, where $Tx_I$ executes first and renders $V$ non-transferable, and (ii) \emph{evade}, where $Tx_E$ executes first and moves $V$ before any freeze can take effect.
If neither side participates, the contest is degenerate and the state does not change in that block.
The issuer values avoiding expected regulatory loss (not seizing $V$), the evader values evasion by $V$, and proposers capture MEV revenue induced by $(b_I,b_E)$.

\subsection{Formal Tullock Contest Model}
\label{subsec:tullock-model}
We study a single sanction contest over balance $V>0$ between issuer $I$ and evader $E$; proposers and MEV infrastructure mediate execution and receive contest expenditures.

\paragraph{Valuations.}
If evasion succeeds, the issuer incurs an expected regulatory loss $\Psi>0$, so $\Psi$ is $I$'s valuation for winning (freeze). The evader values winning by $V$.

\begin{assumption}[Strong regulatory penalties]
\label{ass:APsi}
$\Psi\geq2V$.    
\end{assumption}

This bound is conservative: under IEEPA, the statutory
civil penalty cap is the greater of a fixed amount or
\emph{twice the transaction value}.
\footnote{Civil penalties under 50 U.S.C.\ \S1705 and 31 C.F.R.\ 
\S560.701(a)(2) reach \$377{,}700 (2025), with criminal fines under 
18 U.S.C.\ \S3571 capped at twice the gain or loss. The GENIUS 
Act~\cite{genius-act-2025} classifies stablecoin issuers as Bank 
Secrecy Act financial institutions, triggering direct IEEPA 
liability, and mandates technical freeze/burn capability, making a 
failed freeze independently sanctionable. Reputational harm, 
license-revocation risk, and compounding violations drive 
effective~$\Psi$ well above~$2V$.}

\paragraph{Technology and costs.}
After an off-chain trigger, each player $k\in\{I,E\}$ decides whether to participate (paying a fixed cost $C_k\ge 0$) and, conditional on entry, chooses an expenditure $b_k\ge 0$ via gas/tips and private MEV side payments.

\begin{assumption}[Competitive MEV access]
\label{ass:A1}
    $I$ and $E$ can route $Tx_I$ and $Tx_E$ to an overlapping, economically meaningful set of proposers (directly or via builders); neither has an exclusive channel.
\end{assumption}

We model the issuer's marginal bidding expenditure~$b_I$ as all-pay and the
evader's marginal bidding expenditure~$b_E$ as conditional-pay. The
issuer-side choice is motivated by institutional rather than technical
constraints: while the issuer may choose private or public submission channels,
its compliance role requires it to attempt enforcement independent of the
expected frozen amount. Accordingly, conditional on entry, the issuer's
marginal expenditure is functionally all-pay. By contrast, modern MEV
infrastructure enables evaders to make proposer-facing marginal payments
contingent on execution success (e.g., via atomic bundle execution in
Flashbots), consistent with the private-order-flow patterns documented
in Section~\ref{sec:semev}. Under the single-block contest formulation, the evader
can make at most one effective submission attempt in the contested ordering
game, so the cost of execution failure (e.g., reverted gas consumption) is
absorbed into the fixed participation cost~$C_E$ on a per-contest basis; this
is a conservative approximation in that it understates per-attempt variable
costs and therefore, if anything, overstates the evader's net payoff.

\Paragraph{Timing within a block.}
A single contest associated with one block proceeds as:
\begin{enumerate}[label=\textup{(\roman*)},leftmargin=2.2em,itemsep=2pt]
\item (\emph{Off-chain trigger.}) A public risk signal indicates that the address holding $V$ is likely to be sanctioned.
\item (\emph{Participation.}) Each player decides whether to participate, pays $C_k$ if so, and prepares $Tx_I$ or $Tx_E$.
\item (\emph{Bidding.}) Participants simultaneously choose $b_I,b_E\ge 0$.
\item (\emph{Resolution.}) Proposers receive all relevant bundles; $(P_I,P_E)$ determines whether freeze or evade is realized first. If freeze wins, $V$ is frozen; otherwise $V$ is evaded.
\end{enumerate}

This defines the \emph{SE-MEV contest} for a single block: a two-player Tullock$(r)$ contest between $I$ and $E$ over prize pair $(\Psi,V)$, with contest expenditures accruing to proposers as MEV revenue (or dissipated as contest effort, depending on implementation).

\subsection{Equilibrium of a Single SE-MEV Contest}

Consider a single ordering contest over balance \(V>0\). After entry,
issuer \(I\) and evader \(E\) choose \(b_I,b_E\), with
$
P_I=\frac{b_I^r}{b_I^r+b_E^r},
~
P_E=\frac{b_E^r}{b_I^r+b_E^r}.
$

The evader receives \(V\) only if evasion succeeds and pays \(b_E\) only
upon success, while the issuer pays \(b_I\) regardless of outcome:
$$
U_E=P_E(V-b_E)-C_E,
~
U_I=P_I\Psi-b_I-C_I .
$$

Fixed costs \(C_I,C_E\) do not affect interior bids conditional on entry.

\Cref{prop:1} gives the unique interior equilibrium
\((b_I^*,b_E^*)\) and winning probabilities \((P_I^*,P_E^*)\). By
\Cref{cor:P_I}, stronger penalties shift success toward the issuer:
\(P_I^*>P_E^*\) and \(P_I^*\to 1\) as \(\Psi/V\to\infty\).

The cost is ordering expenditure. By \Cref{cor:1},
$
T^*:=b_I^*+P_E^*b_E^*
$
is extracted by proposers and MEV intermediaries. Thus, stronger penalties
make enforcement more likely, but make the priority race arbitrarily costly
as \(\Psi/V\) grows.
\subsection{Entry Under Strong Penalties}

Under the sufficient-penalty condition in \Cref{cor:observ1}, issuer
participation is individually rational. If the issuer does not enter, its
payoff is normalized to zero, but an active evader
can bid \(b_E=\varepsilon>0\) and win. By \Cref{cor:observ1}, $
    U_I^*
    =
    P_I^*\Psi-b_I^*-C_I
    >
    0.$ 
Thus, sufficiently strong penalties make issuer entry and positive ordering-market bidding strictly optimal.

\subsection{Fixed Costs and Evasion Specialization}
Fixed participation costs create economies of scale for evasion.
Consider $n$ evaders $A_1,\ldots,A_n$ with balances $V_i$ and
issuer penalty exposures $\Psi_i$, where
$
\sum_{i=1}^n V_i = V, ~
\sum_{i=1}^n \Psi_i = \Psi, ~
\frac{\Psi_i}{V_i}=\frac{\Psi}{V}.
$
An evader can either contest individually, paying a fixed cost
$C_E^{\mathrm{solo}}$, or delegate execution to a specialized bot $S$,
which aggregates volume and pays a fixed cost $C_S$ only once. Under solo evasion, evader $A_i$ obtains
\[
U_{A_i}^{\mathrm{solo}}
=
P_E^*(\Psi_i,V_i)
\left(
V_i-b_E^*(\Psi_i,V_i)
\right)
-
C_E^{\mathrm{solo}}
=
\frac{V_i}{1+(1+r)(s^*)^r}
-
C_E^{\mathrm{solo}} .
\]
Since the gross term is linear in $V_i$, sufficiently small evaders
obtain negative solo payoff. Under delegation, the bot's gross contest surplus is
\[
U_S^{\mathrm{gross}}
=
P_E^*(\Psi,V)
\left(
V-b_E^*(\Psi,V)
\right).
\]
If the bot pays $C_S$ once and charges commission $f\in(0,1)$ on
net contest surplus, evader $A_i$'s delegated payoff is
\[
U_{A_i}^{\mathrm{delegate}}
=
(1-f)\frac{V_i}{V}
\left(
U_S^{\mathrm{gross}}-C_S
\right).
\]
Thus delegation can be attractive for small evaders precisely because
the fixed cost is amortized across pooled volume.
\begin{observation}
[Centralized evasion bots] With non-trivial
fixed costs \(C_E\), small evaders optimally delegate to a few specialized
bots that amortize \(C_E\) across large illicit volume; the issuer therefore competes against few sophisticated MEV-aware
adversaries
\end{observation}
\subsection{Vertical Integration and Ordering Capture}

Let \(\gamma\in[0,1]\) denote the issuer's ordering-control level: with
probability \(\gamma\), the freeze obtains priority; otherwise the race enters
the open SE-MEV contest. The issuer's expected cost is
\[
  \mathcal{C}(\gamma;\Psi,V)
  =
  C_I+(1-\gamma)M(\Psi,V)+\kappa(\gamma),
\]
where \(M(\Psi,V)\) is the issuer's open-contest exposure and \(\kappa(\gamma)\)
is the cost of ordering control.

By \Cref{cor:2}, \(M(\Psi,V)\to\infty\) as \(\Psi/V\to\infty\). Stronger
regulatory penalties therefore make open contests increasingly costly,
increasing the issuer's incentive to raise \(\gamma\), i.e., to seek greater
ordering control through vertical integration or ordering capture.

\section{Discussion}
\label{sec:discussion}
We discuss insights of SE-MEV for the blockchain ecosystems.

\noindent\textbf{Why Issuer-Side Hardening Is Not Enough.} SE-MEV is not a narrow implementation flaw, but a structural limitation of reactive sanctions on public blockchains. Sanction authority is exercised through ordinary transactions, while enforcement depends on transaction ordering. Whenever a freeze transaction and an evading transfer enter the same ordering pipeline, enforcement becomes a priority contest rather than a purely issuer-controlled action.

Issuer-side hardening can reduce operational delay, but cannot remove this contest. The USDT multisignature workflow amplifies the race by separating authorization from execution, whereas a USDC-style single-transaction design removes this extra delay. Yet in both cases, the freeze still competes for inclusion and ordering. Partial compliance by builders or proposers is also insufficient, because heterogeneous order flow, private channels, and synchronization gaps leave enforcement exposed to profit-driven decisions.

Regulation further raises the stakes in this race. By requiring permitted issuers to maintain freeze-or-burn capabilities and comply with lawful blocking orders, the GENIUS Act increases the cost of enforcement failure and, in our model, raises the issuer's effective downside risk $\Psi$. This strengthens issuers' incentives to pay for ordering priority, amplifying the MEV-tax escalation predicted by Corollary~1. Robust mitigation therefore likely requires protocol or infrastructure-level support, such as trusted sequencing~\cite{kelkar2023themis,park2025frontrunning, kelkar2020order}, confidential or TEE-assisted ordering~\cite{zhang2016town}, or other mechanisms~\cite{choudhuri2025practical, agarwal2025efficiently, bowe2020zexe, kosba2016hawk} that more tightly couple privileged enforcement actions with execution priority. The key point is that SE-MEV reflects a structural limit of public-chain enforcement, not a transient inefficiency in current issuer implementations.

\noindent\textbf{Implications for Decentralization.}
SE-MEV reveals a structural threat to blockchain decentralization under regulatory pressure. To ensure sanctions compliance, stablecoin issuers are pushed to compete for transaction ordering priority and, given their scale, capital, and operational capacity, are well positioned to become dominant bidders or deeply integrated participants in the MEV supply chain. Because these issuers are centralized and directly regulated, this creates an indirect channel for state influence over transaction ordering, without requiring protocol-level intervention. As SE-MEV incentives intensify, ordering decisions may increasingly reflect the objectives of a small set of regulated actors, weakening the decentralized security assumptions of permissionless blockchains.

\noindent\textbf{Protocol-Dependent Manifestation.}
Although our empirical analysis focuses on USDT and USDC on Ethereum, the core SE-MEV conflict is not limited to these assets or to Ethereum. It can arise for any centralized stablecoin (e.g., PYUSD, TUSD) that enforces sanctions through reactive, on-chain mechanisms, across other Layer-1 and Layer-2 systems where transaction ordering is controlled by economically motivated actors. 
The same incentive conflict, however, can lead to different outcomes under different ordering designs. 
On Ethereum, where ordering is mediated by an open MEV market, SE-MEV is likely to appear as market capture, pushing issuers toward aggressive bidding or tighter integration with MEV infrastructure. 
On systems such as Tron, where ordering power is concentrated among a small set of validators and public auctions are absent, the same pressure may instead favor off-chain coordination with ordering authorities. 
SE-MEV is therefore general in cause but protocol-dependent in form, with its decentralization impact shaped by each chain's ordering design.

\noindent\textbf{Temporal Coverage.}
Our dataset ends on August 31, 2025 and does not capture the most recent months of on-chain activity. The nearly eight-year span nonetheless covers the Proof-of-Work to Proof-of-Stake transition, the maturation of the MEV supply chain, and the period following the GENIUS Act's enactment in July 2025, across which the same SE-MEV escalation trajectory persists.

\section{Related Works}

Recent work has quantitatively examined the impact of regulation and sanctions on blockchain systems. Zola et al.~\cite{zola2024assessing} analyze the effectiveness of financial sanctions on sanctioned entities in the Bitcoin ecosystem. Wahrstätter et al.~\cite{wahrstatter2024blockchain} study blockchain censorship under regulatory pressure, analyzing how OFAC sanctions affect transaction inclusion at the consensus and infrastructure layers. Liu et al.~\cite{liu2025evasion} investigate the effectiveness of blockchain sanctions using Tornado Cash as a case study on Ethereum, focusing on post-sanction fund propagation under enforcement. Practitioner reports have also provided USDT statistics on blacklisting and fund flows, but they do not systematically analyze sanction evasion itself and the underlying ordering-layer race~\cite{blocksec_frozon}.

The earliest empirical studies on Maximal Extractable Value (MEV) were conducted by Daian et al. \cite{daian2020flash}, with subsequent contributions from Torres et al. \cite{torres2021frontrunner} and Qin et al. \cite{qin2022quantifying}. These works laid the foundation for understanding the prevalence and impact of MEV in blockchain ecosystems. Later, Weintraub et al. \cite{weintraub2022flash} applied similar methods to measure the occurrence of MEV within private pools, while Li et al. \cite{li2023demystifying} specifically examined MEV activities within Flashbots bundles. From a theoretical perspective, Mazorra et al. \cite{mazorra2022price} employed game theory to analyze the strategies of MEV participants, offering a deeper understanding of their behavior in the market. Additionally, various works have explored mitigation strategies for MEV at the DeFi application layer~\cite{ma2025surviving, milionis2024automated, zhang2025maximal} or consensus layer~\cite{bormet2025beat}, aiming to reduce its negative impacts on fairness and security. However, none of these studies has identified MEV arising from sanctions evasion.

\bibliographystyle{plain}
\bibliography{ref}
\clearpage
\clearpage
\appendix
\setcounter{table}{0}
\renewcommand{\thetable}{A\arabic{table}} %
\section{Formal Game-Theoretic Analysis and Proofs}
\begin{lemma}\label{lem:PhiMonotone}
Fix $r\ge 1$ and define
$
\Phi(s)=\frac{s(1+s^r)^2}{1+(1+r)s^r},
\qquad s\ge 1.
$
Then $\Phi$ is strictly increasing on $[1,\infty)$.
\end{lemma}

\begin{proof}
Let $x=s^r$ and write
$
\widetilde{\Phi}(x)\equiv \frac{x^{1/r}(1+x)^2}{1+(1+r)x}.
$
Since $x=s^r$ is strictly increasing for $s\ge 1$, it suffices to show that
$\widetilde{\Phi}$ is increasing on $[1,\infty)$.

A direct calculation gives
$
\frac{d}{dx}\log \widetilde{\Phi}(x)
=
\frac{Q(x)}{rx(1+x)\bigl(1+(1+r)x\bigr)},
$
where
$
Q(x)=1+(2+2r-r^2)x+(1+r)^2x^2.
$
The denominator is positive for $x\ge 1$. Moreover,
$
Q'(x)=2+2r-r^2+2(1+r)^2x\ge r^2+6r+4>0
\qquad (x\ge 1),
$
and
$
Q(1)=4(r+1)>0.
$
Hence $Q(x)>0$ for all $x\ge 1$, so
$\frac{d}{dx}\log \widetilde{\Phi}(x)>0$ and therefore
$\widetilde{\Phi}$, and thus $\Phi$, is strictly increasing on $[1,\infty)$.
\end{proof}

\begin{proposition}[Equilibrium bids in the SE-MEV contest]
\label{prop:1}
Consider a single sanction contest in which both \(I\) and \(E\)
participate and play pure strategies. Under \Cref{ass:APsi}, there exists a unique interior
Nash equilibrium with bids \((b_I^*,b_E^*)\) and success probabilities
\((P_I^*,P_E^*)\).
\end{proposition}

\begin{proof}[Proof of \Cref{prop:1}]
Condition on entry and suppress \(C_I,C_E\). For \(b_I+b_E>0\),
$
P_I=\frac{b_I^r}{b_I^r+b_E^r},\qquad
P_E=\frac{b_E^r}{b_I^r+b_E^r}.
$
Any interior equilibrium satisfies
\begin{equation}
\label{eq:FOC}
b_I^*=\Psi rP_I^*P_E^*,
\qquad
b_E^*=\frac{rP_I^*V}{1+rP_I^*}.
\end{equation}
Let \(s^*=b_I^*/b_E^*\). Since
$
P_I^*=\frac{(s^*)^r}{1+(s^*)^r},
P_E^*=\frac{1}{1+(s^*)^r},
$
dividing the two equations in \eqref{eq:FOC} gives
\begin{equation}
\label{eq:star}
\Phi(s^*)=\frac{\Psi}{V},\qquad
\Phi(s):=\frac{s(1+s^r)^2}{1+(1+r)s^r}.
\end{equation}

There is no solution on \([0,1]\). Indeed, with \(x=s^r\),
$
\Phi(s)\le \frac{(1+x)^2}{1+(1+r)x}<2
\qquad (s\in[0,1]),
$
because
$
2(1+(1+r)x)-(1+x)^2=1+2rx-x^2>0.
$
Since \(\Psi/V\ge2\), any solution must satisfy \(s>1\). On
\([1,\infty)\), \(\Phi\) is strictly increasing by
\Cref{lem:PhiMonotone}, while \(\Phi(1)<2\le \Psi/V\) and
\(\Phi(s)\to\infty\). Hence \eqref{eq:star} has a unique solution
\(s^*>1\), which uniquely determines \((b_I^*,b_E^*)\) through
\eqref{eq:FOC}.

It remains only to verify best responses. For the evader, fixing
\(b_I>0\) and writing \(\tilde s=b_E/b_I\),
\[
\frac{d}{d\tilde s}U_E(b_I,b_E)
=
\frac{b_I \tilde s^{r-1}}{(1+\tilde s^r)^2}
\left(\frac{rV}{b_I}-(r+1)\tilde s-\tilde s^{r+1}\right).
\]
The bracketed term is strictly decreasing from a positive value to
\(-\infty\), so the evader has a unique interior best response.

For the issuer, fixing \(b_E>0\) and writing \(s=b_I/b_E\),
\[
U_I=\Psi\frac{s^r}{1+s^r}-b_Es,
\qquad
\frac{dU_I}{ds}
=
\Psi\frac{rs^{r-1}}{(1+s^r)^2}-b_E.
\]
For \(r=1\), this derivative is strictly decreasing. For \(r>1\), the
first term is single-peaked with peak at
\[
s_0=\left(\frac{r-1}{r+1}\right)^{1/r}<1,
\]
so the candidate \(s^*>1\) is the unique right-branch maximizer. Finally,
\((s^*)^r>r-1\); otherwise \eqref{eq:star} would imply
\[
\frac{\Psi}{V}
\le s^*\le (r-1)^{1/r}<2,
\]
contradicting \(\Psi/V\ge2\). Thus
\[
U_I(b_I^*,b_E^*)
=
\Psi\frac{(s^*)^r(1+(s^*)^r-r)}{(1+(s^*)^r)^2}
>0
=
U_I(0,b_E^*),
\]
so the issuer also strictly prefers the candidate to the boundary
\(b_I=0\). Therefore \((b_I^*,b_E^*)\) is the unique interior Nash
equilibrium.
\end{proof}
\begin{corollary}[Enforcement advantage]\label{cor:P_I}
Under \Cref{ass:APsi},
    \begin{enumerate}[label=\textup{(\roman*)},leftmargin=2.2em,itemsep=2pt]
        \item \(P_I^*>1/2.\)
        \item As \(\Psi/V\) rises, \(P_I^*\)
        strictly increases and converges to \(1\).
    \end{enumerate}
\end{corollary}
\begin{proof}
By \Cref{prop:1}, \(s^*>1\). Hence
\[
P_I^*
=
\frac{(s^*)^r}{1+(s^*)^r}
>
\frac12
>
P_E^* .
\]

Moreover, \(s^*\) is pinned down by
$
\Phi(s^*)=\frac{\Psi}{V}.
$
By \Cref{lem:PhiMonotone}, \(\Phi\) is strictly increasing on
\([1,\infty)\). Thus stricter regulation, i.e. a larger \(\Psi/V\),
implies a larger \(s^*\), and therefore a larger
$
P_I^*=\frac{(s^*)^r}{1+(s^*)^r}.
$
Since \(\Phi(s)\to\infty\), we have \(s^*\to\infty\) as
\(\Psi/V\to\infty\), so \(P_I^*\to1\).
\end{proof}

\begin{corollary}[The MEV tax]
\label{cor:1}
In SE-MEV equilibrium, block proposers and MEV intermediaries extract an
implicit \emph{MEV tax} equal to the expected aggregate ordering
expenditure \(T^*\).
\Cref{prop:1} implies that
\begin{enumerate}[label=\textup{(\roman*)},leftmargin=2.2em,itemsep=2pt]
    \item $
    T^*/V>\frac{r}{r+2}.
    $
    \item $
    T^*/V\to\infty
    \qquad\text{as}\qquad
    \Psi/V\to\infty.
    $
\end{enumerate}
\end{corollary}

\begin{proof}
In equilibrium,
\[
T^*:=b_I^*+P_E^*b_E^*
=
b_E^*(s^*+P_E^*).
\]
The term \(b_I^*\) is counted in full because the issuer's expenditure is
all-pay, while \(b_E^*\) is weighted by \(P_E^*\) because the evader's
expenditure is paid only upon successful evasion.

By \Cref{cor:P_I}, \(P_I^*>1/2\). Since
$
b_E^*
=
\frac{rP_I^*V}{1+rP_I^*},
$
we have
$
b_E^*>\frac{r}{r+2}V.
$
Together with \(s^*>1\), this implies
\[
T^*=b_E^*(s^*+P_E^*)>b_E^*>\frac{r}{r+2}V.
\]

Finally, since \(\Phi(s^*)=\Psi/V\), \(\Phi\) is strictly increasing, and
\(\Phi(s)\to\infty\), stricter regulation implies \(s^*\to\infty\).
Hence \(P_I^*\to1\). Hence
\[
\frac{b_E^*}{V}
=
\frac{rP_I^*}{1+rP_I^*}
\to
\frac{r}{1+r},
\]
while \(s^*+P_E^*\to\infty\). Therefore
\[
\frac{T^*}{V}
=
\frac{b_E^*}{V}(s^*+P_E^*)
\to\infty.
\]
\end{proof}

\begin{corollary}[Issuer entry under strong penalties]
\label{cor:observ1}
Normalize the issuer's payoff from non-participation to zero. For any
finite \(C_I\), there exists \(\bar\lambda_I(C_I/V,r)<\infty\) such that,
whenever \(\Psi/V>\bar\lambda_I\),
\[
U_I^*=P_I^*\Psi-b_I^*-C_I>0 .
\]
Hence, sufficiently strong regulatory penalties make entry into the
ordering contest strictly preferable to non-participation.
\end{corollary}

\begin{proof}
Let \(\lambda=\Psi/V\) and \(G_I^*=P_I^*\Psi-b_I^*\). By
\Cref{prop:1},
\[
\frac{G_I^*}{V}
=
\lambda
\frac{(s^*)^r\bigl(1+(s^*)^r-r\bigr)}
{\bigl(1+(s^*)^r\bigr)^2},
\qquad
\Phi(s^*)=\lambda .
\]
Since \(\Phi\) is increasing and unbounded, \(\lambda\to\infty\) implies
\(s^*\to\infty\), and the fraction above converges to \(1\). Thus
\(G_I^*/V\to\infty\). For any finite \(C_I/V\), choose
\(\bar\lambda_I\) so that \(G_I^*/V>C_I/V\) whenever
\(\lambda>\bar\lambda_I\). Then \(U_I^*=G_I^*-C_I>0\).
\end{proof}

The parameter \(\gamma\) is reduced-form and may capture builder integration,
validator relationships, privileged relay access, trusted sequencing, or other
ordering arrangements; it need not correspond to literal validator ownership.
By construction, \(\gamma\) changes only the probability of entering the open
SE-MEV contest. Conditional on entry, equilibrium bids and win probabilities
remain those in \Cref{prop:1}.

Let
\[
M(\Psi,V):= b_I^*+P_E^*\Psi ,
\]
where \(b_I^*\) is the issuer's equilibrium priority spending and
\(P_E^*\Psi\) is its expected regulatory loss from successful evasion. The
issuer's expected cost under ordering control is therefore
\[
  \mathcal{C}(\gamma;\Psi,V)
  =
  C_I+(1-\gamma)M(\Psi,V)+\kappa(\gamma).
\]
We assume \(\kappa\in C^1([0,1))\) with \(\kappa(0)=0\), and use the continuous
extension of \(\mathcal{C}\) to \(\gamma=1\) whenever the limit exists.

\begin{corollary}[Enforcement-side value of ordering control]
\label{cor:2}
Fix \(V>0\) and the SE-MEV equilibrium of \Cref{prop:1}. 
\begin{enumerate}[label=\textup{(\roman*)},leftmargin=2.2em,itemsep=2pt]
  \item \(E_I\) is strictly decreasing in \(\gamma\)
  with slope \(-M(\Psi,V)<0\), and \(M(\Psi,V)/V\to\infty\) as
  \(\Psi/V\to\infty\).

  \item For any \(\gamma_0\in[0,1)\),
  $
  \partial_\gamma\mathcal{C}(\gamma;\Psi,V)|_{\gamma=\gamma_0}<0
  ~\text{iff}~
  \kappa'(\gamma_0)<M(\Psi,V).
$

  \item If \(\kappa'\) is nondecreasing, then for every
  \(\bar\gamma\in(0,1)\) there exists
  \(\bar\theta(\bar\gamma,V)<\infty\) such that
  \(\Psi/V>\bar\theta(\bar\gamma,V)\) implies
  \[
  \partial_\gamma\mathcal{C}(\gamma;\Psi,V)<0
  \quad
  \forall \gamma\in[0,\bar\gamma].
  \]

  \item
If \(\kappa\) extends continuously to \([0,1]\), then every minimizer
\[
\gamma^*(\Psi,V)\in
\arg\min_{\gamma\in[0,1]}\mathcal{C}(\gamma;\Psi,V)
\]
satisfies
\[
\gamma^*(\Psi,V)\to1
\qquad
\text{as }\Psi/V\to\infty
\]
for fixed \(V>0\).
\end{enumerate}
\end{corollary}
\begin{proof}
By effective ordering control,
\[
E_I(\gamma;\Psi,V)
=
C_I+(1-\gamma)M(\Psi,V),
\qquad
M(\Psi,V)=b_I^*+P_E^*\Psi .
\]

\textup{(i)}
Differentiating gives
$
\partial_\gamma E_I=-M(\Psi,V)<0,$ so \(E_I\) is strictly decreasing in \(\gamma\).

Moreover, by \Cref{prop:1}, \(b_E^*=\frac{rP_I^*V}{1+rP_I^*}<V\).
Since \(\Psi\ge2V\), we have \(\Psi>b_E^*\), and hence
\[
M=b_I^*+P_E^*\Psi>b_I^*+P_E^*b_E^*=T^* .
\]
Thus \(M/V>T^*/V\to\infty\) by \Cref{cor:1}.

\textup{(ii)}
Since
$
\partial_\gamma\mathcal{C}
=
\kappa'(\gamma)-M(\Psi,V),
$
the claim follows immediately.

\textup{(iii)}
Fix \(\bar\gamma\in(0,1)\). By \textup{(i)}, for fixed \(V>0\),
\(M(\Psi,V)\to\infty\) as \(\Psi/V\to\infty\). Hence there exists
\(\bar\theta(\bar\gamma,V)<\infty\) such that
\[
M(\Psi,V)>\kappa'(\bar\gamma)
\]
whenever \(\Psi/V>\bar\theta(\bar\gamma,V)\). Since \(\kappa'\) is
nondecreasing, \(\kappa'(\gamma)\le\kappa'(\bar\gamma)\) for all
\(\gamma\in[0,\bar\gamma]\), so
\[
\partial_\gamma\mathcal{C}(\gamma;\Psi,V)
=
\kappa'(\gamma)-M(\Psi,V)<0 .
\]

\textup{(iv)}
Assume \(\kappa\) extends continuously to \([0,1]\), and let
\[
\gamma^*(\Psi,V)\in
\arg\min_{\gamma\in[0,1]}\mathcal{C}(\gamma;\Psi,V).
\]
Fix \(V>0\) and any \(\varepsilon\in(0,1)\). For
\(\gamma\in[0,1-\varepsilon]\),
\[
\mathcal{C}(\gamma;\Psi,V)-\mathcal{C}(1;\Psi,V)
\ge
\min_{\tilde\gamma\in[0,1-\varepsilon]}
\bigl(\kappa(\tilde\gamma)-\kappa(1)\bigr)+\varepsilon M(\Psi,V)>-\infty
\]
by continuity of \(\kappa\). Since \(M(\Psi,V)\to\infty\) as
\(\Psi/V\to\infty\), the right-hand side is positive for sufficiently
large \(\Psi/V\). Hence no minimizer lies in \([0,1-\varepsilon]\), so
\(\gamma^*(\Psi,V)>1-\varepsilon\) eventually. Because
\(\varepsilon\) is arbitrary,
\[
\gamma^*(\Psi,V)\to1 .
\]
\end{proof}

\section{Robustness to the Outflow Ratio Threshold $\alpha$}
\label{app:alpha-robustness}

We assess how the choice of the outflow ratio threshold $\alpha$
affects the set of evasion sessions identified by the procedure of
Section~\ref{sec:session-construction}. The minimum outflow threshold
is fixed at $\beta = \$1{,}000$. Out of $110{,}199$ candidate
sessions, Table~\ref{tab:alpha-counts} reports the number retained,
the number of distinct addresses, and the change in retained
sessions relative to $\alpha = 10\%$.

\begin{table}[h]
\centering
\caption{Retained sessions and addresses across $\alpha$, with
$\beta = \$1{,}000$ fixed.}
\begin{tabular}{rrrrr}
\toprule
$\alpha$ & Retained & Retention & Addresses & $\Delta$ vs.\ 10\% \\
\midrule
1\%   & 25{,}399 & 23.05\% & 884 & $+37.8\%$ \\
5\%   & 20{,}952 & 19.01\% & 856 & $+13.7\%$ \\
10\%  & 18{,}432 & 16.73\% & 836 & --- \\
30\%  & 13{,}475 & 12.23\% & 779 & $-26.9\%$ \\
50\%  & 10{,}745 &  9.75\% & 719 & $-41.7\%$ \\
70\%  &  8{,}819 &  8.00\% & 676 & $-52.2\%$ \\
90\%  &  6{,}955 &  6.31\% & 637 & $-62.3\%$ \\
\bottomrule
\end{tabular}

\label{tab:alpha-counts}
\end{table}

Table~\ref{tab:alpha-amounts} summarises the outflow distribution
and the realised outflow ratio within each retained set.

\begin{table}[h]
\centering
\caption{Outflow size and realised ratio within retained sessions.}
\resizebox{1\linewidth}{!}{
\begin{tabular}{rrrrr}
\toprule
$\alpha$ & Median (USD) & Mean (USD) & Median ratio & Mean ratio \\
\midrule
1\%  & 43{,}000 & 166{,}097 & 35.13\% & 46.46\% \\
5\%  & 55{,}810 & 194{,}854 & 52.12\% & 55.77\% \\
10\% & 65{,}000 & 209{,}587 & 65.52\% & 62.40\% \\
30\% & 77{,}410 & 239{,}672 & 91.72\% & 78.55\% \\
50\% & 80{,}000 & 252{,}998 & 97.96\% & 88.46\% \\
70\% & 80{,}080 & 257{,}481 & 99.46\% & 94.78\% \\
90\% & 79{,}133 & 217{,}464 & 99.95\% & 98.62\% \\
\bottomrule
\end{tabular}
\label{tab:alpha-amounts}
}
\end{table}

Three observations follow. First, the count of retained sessions is
moderately sensitive to $\alpha$, falling by $62.3\%$ as $\alpha$
rises from $10\%$ to $90\%$, while the count of distinct addresses
falls by only $23.8\%$. The marginal sessions removed by stricter
$\alpha$ are concentrated within addresses that already contribute
multiple sessions, so the population of suspect addresses is stable
across the range tested. Second, the outflow size distribution is
nearly invariant for $\alpha \geq 30\%$: the median lies in
$[\$77.4\text{k}, \$80.1\text{k}]$ and the mean in
$[\$239.7\text{k}, \$257.5\text{k}]$, with variation below $5\%$.
Raising $\alpha$ removes small or partial transfers but leaves the
high-value behaviour unchanged. Third, the realised outflow ratio
increases monotonically with $\alpha$, and its gap to the threshold
narrows from $+55.5$ percentage points at $\alpha = 10\%$ to
$+9.95$ at $\alpha = 90\%$, confirming that retained sessions
approach a near-complete sweep as $\alpha$ tightens.

The transition between $\alpha = 10\%$ and $\alpha = 30\%$ accounts
for most of the change in behavioural composition (median ratio
$65.5\% \to 91.7\%$); beyond $\alpha = 50\%$ all reported quantities
plateau. We adopt $\alpha = 10\%$ as the baseline for breadth of
coverage, and report $\alpha = 50\%$ as a sensitivity check focused
on near-complete liquidations.

\section{Analysis of Gas Ratio of Private and Public Transactions }
\label{app:private_tx}

\begin{table}[htbp]
\centering
\caption{Quantile Comparison of Gas Ratio of Private and Public Transactions }
\label{tab:quantile}
\begin{tabular}{lccc}
\toprule
Quantile & Non-private & Private & priv/pub \\
\midrule
P5  & 0.830 & 0.903 & 1.09$\times$ \\
P10 & 0.890 & 0.944 & 1.06$\times$ \\
P25 & 0.952 & 0.984 & 1.03$\times$ \\
\textbf{P50 (median)} & \textbf{0.991} & \textbf{1.040} & 1.05$\times$ \\
\textbf{P75} & 1.093 & \textbf{1.545} & \textbf{1.41$\times$} \\
\textbf{P90} & 1.340 & \textbf{3.074} & \textbf{2.29$\times$} \\
\textbf{P95} & 1.685 & \textbf{3.963} & 2.35$\times$ \\
P99 & 5.583 & 7.573 & 1.36$\times$ \\
\textbf{Mean} & 1.202 & \textbf{1.587} & \textbf{1.32$\times$} \\
\bottomrule
\end{tabular}
\end{table}

As Table~\ref{tab:quantile} shown, the medians differ by only 5\%, yet the right tails diverge sharply: a 41\% gap at P75 and 129\% at P90. This explains why median-based comparisons obscure the underlying difference in distributional shape.


\begin{table}[h]

\centering
\caption{Tail Concentration}
\resizebox{1\linewidth}{!}{
\label{tab:concentration}
\begin{tabular}{lccc}
\toprule
Effective Gas Ratio Threshold & Non-private & Private & $\Delta$ \\
\midrule
$\geq 1.0$ (above block median)   & 43.3\% & \textbf{64.9\%} & +21.6 pp \\
$\geq 1.5$ ($\geq$50\% premium)   & 6.1\%  & \textbf{27.0\%} & \textbf{+21.0 pp} \\
$\geq 2.0$ ($\geq$2$\times$)      & 4.0\%  & \textbf{14.9\%} & +10.9 pp \\
$\geq 3.0$ ($\geq$3$\times$)      & 3.3\%  & \textbf{10.8\%} & +7.5 pp \\
\bottomrule
\end{tabular}
}
\end{table}

We analyze the tail concentration. Table~\ref{tab:concentration} shows over a quarter of private transactions pay at least 1.5$\times$ the block-median gas price:4.4$\times$ the rate observed in non-private transactions (6.1\%).

\section{Stake-Sensitive Public Fee Competition}

\label{app:fee-stake}
\begin{table*}[htbp]
\centering
\caption{The Relationship between episode-level stake and public fee premia across pre-sanction time windows. $L_{\text{episode}}$ proxies the value controllable within an evasion episode. Public fee premia are measured as the ratio between fees actually paid and counterfactual fees computed at the block-median gas price. Positive associations appear in short-horizon windows but weaken in the broader tactical-reactive regime.}

\label{tab:fee-stake}
\begin{tabular}{l l r r r r}

\toprule
Window & Interpretation & $n$ & Spearman $\rho$ & OLS $\hat{\beta}$ & $p$-value \\
\midrule
$\Delta \leq 242$s & Canonical race & 9 & 0.4407 & -- & -- \\
$\Delta \leq 1200$s & Relaxed short horizon & 20 & 0.4910 & 0.2317 & 0.0017 \\
$\Delta \leq 3600$s & Relaxed short horizon & 31 & 0.1150 & 0.1366 & 0.0226 \\
$242\text{s} < \Delta \leq 86400\text{s}$ & Tactical-reactive & 456 & -0.0387 & -0.0172 & 0.0594 \\
\bottomrule
\end{tabular}
\end{table*}
\end{document}